\documentclass[number]{article}
\usepackage[latin9]{inputenc}
\usepackage{float}
\usepackage{amsmath}
\usepackage{amsthm}
\usepackage{graphicx}
\usepackage{setspace}
\usepackage[numbers]{natbib}

\makeatletter

\DeclareTextSymbolDefault{\textquotedbl}{T1}
\providecommand{\tabularnewline}{\\}
\floatstyle{ruled}
\newfloat{algorithm}{tbp}{loa}
\providecommand{\algorithmname}{Algorithm}
\floatname{algorithm}{\protect\algorithmname}

\newcommand{\lyxaddress}[1]{
	\par {\raggedright #1
	\vspace{1.4em}
	\noindent\par}
}
\theoremstyle{plain}
\newtheorem{thm}{\protect\theoremname}
\theoremstyle{plain}
\newtheorem{lem}[thm]{\protect\lemmaname}
\ifx\proof\undefined
\newenvironment{proof}[1][\protect\proofname]{\par
	\normalfont\topsep6\p@\@plus6\p@\relax
	\trivlist
	\itemindent\parindent
	\item[\hskip\labelsep\scshape #1]\ignorespaces
}{%
	\endtrivlist\@endpefalse
}
\providecommand{\proofname}{Proof}
\fi

\usepackage{algpseudocode}
\usepackage{tikz}
\usepackage{multicol}

\mathchardef\mhyphen="2D
\@ifundefined{showcaptionsetup}{}{%
 \PassOptionsToPackage{caption=false}{subfig}}
\usepackage{subfig}
\makeatother

\providecommand{\lemmaname}{Lemma}
\providecommand{\theoremname}{Theorem}

\begin{document}
\title{Optimal Task Scheduling Benefits From a Duplicate-Free State-Space}
\author{Michael Orr \& Oliver Sinnen}
\maketitle

\lyxaddress{\begin{center}
Department of Electrical and Computer Engineering, University of
Auckland, New Zealand
\par\end{center}}
\begin{abstract}
The NP-hard problem of task scheduling with communication delays ($P|prec,c_{ij}|C_{\mathrm{max}}$)
is often tackled using approximate methods, but guarantees on the
quality of these heuristic solutions are hard to come by. Optimal
schedules are therefore invaluable for properly evaluating these heuristics,
as well as being very useful for applications in time critical systems.
Optimal solving using branch-and-bound algorithms like A{*} has been
shown to be promising in the past, with a state-space model we refer
to as exhaustive list scheduling (ELS). The obvious weakness of this
model is that it leads to the production of large numbers of duplicate
states during a search, requiring special techniques to mitigate this
which cost additional time and memory. In this paper we define a new
state-space model (AO) in which we divide the problem into two distinct
sub-problems: first we decide the allocations of all tasks to processors,
and then we order the tasks on their allocated processors in order
to produce a complete schedule. This two-phase state-space model offers
no potential for the production of duplicates. We also describe how
the pruning techniques and optimisations developed for the ELS model
were adapted or made obsolete by the AO model. An experimental evaluation
shows that the use of this new state-space model leads to a significant
increase in the number of task graphs able to be scheduled within
a feasible time-frame, particularly for task graphs with a high communication-to-computation
ratio. Finally, some advanced lower bound heuristics are proposed
for the AO model, and evaluation demonstrates that significant gains
can be achieved from the consideration of necessary idle time. 
\end{abstract}

\section{Introduction}

In order to use the full potential of a multiprocessor system in speeding
up task execution, efficient schedules are required. In this work,
we address the classic problem of task scheduling with communication
delays, known as $P|prec,c_{ij}|C_{\mathrm{max}}$ using the $\alpha|\beta|\gamma$
notation \citep{Veltman1990:msc}. The problem involves a set of tasks,
with associated precedence constraints and communication delays, which
must be scheduled such that the overall finish time (schedule length)
is minimised. The optimal solving of this problem is well known to
be NP-hard \citep{Sar1989p}, so that the amount of work required
grows exponentially as the number of tasks is increased. For this
reason, many heuristic approaches have been developed, trading solution
quality for reduced computation time \citep{HagJan05AHP,YanGer1993LSW,HWACHO89Sch,Sin07TSP}.
Unfortunately, the relative quality of these approximate solutions
cannot be guaranteed, as no $\alpha$-approximation scheme for the
problem is known \citep{DM:2009:SPP}.

Although the NP-hardness of the problem usually discourages optimal
solving, an optimal schedule can give a significant advantage in time
critical systems or applications where a single schedule is reused
many times. Optimal solutions are also necessary in order to evaluate
the effectiveness of a heuristic scheduling method. Branch-and-bound
algorithms have previously shown promise in efficiently finding optimal
solutions to this problem \citep{ShaSinAst2010}, but the state-space
model used, exhaustive list scheduling (ELS), was prone to the production
of duplicate states.

This paper presents a new state-space model in which the task scheduling
problem is tackled in two distinct phases: first allocation, and then
ordering. The two-phase state-space model (abbreviated AO) does not
allow for the possibility of duplicate states. We give a detailed
explanation of the algorithms and heuristics used to explore the AO
state space and discuss its theoretical benefits over ELS, as well
as its limitations. We describe the algorithms necessary to guarantee
that all valid states can be produced, and that invalid states cannot.
Previous research demonstrated that the ELS model benefited immensely
from the application of pruning techniques such as processor normalisation
and identical task pruning. The benefit is such that effective pruning
would seem to be a necessity in order for another model to be competitive.
We therefore describe these techniques and explain the changes needed
to adapt them from the ELS model to the AO model. An experimental
evaluation is used to compare the performance of the two state-space
models. We then propose further advances to lower bound heuristics
for AO, and evaluate these. This paper expands on preliminary work,
which featured a promising evaluation using only small task graphs
\citep{Orr2015:dfs}.

In Section \ref{sec:background}, background information is given,
including an explanation of the task scheduling model, an overview
of branch-and-bound algorithms, and a description of the ELS model.
Section \ref{sec:theory} describes the new AO model, and how a branch-and-bound
search is conducted through it. Section \ref{sec:Avoiding-Invalid-Stat}
proposes a technique to avoid searching invalid states in the ordering
phase. Section \ref{sec:pruning} explains the pruning techniques
used with ELS, and their adaptation to the AO model. Section \ref{sec:evaluation}
explains how the new model was evaluated by comparison with the old
one, and presents the results. Section \ref{sec:Lower-Bound-Heuristics}
describes novel and more complex lower bound heuristics for the AO
model, and an evaluation of their effect. Finally, Section~\ref{sec:conclusion}
gives the conclusions of the paper and outlines possible further avenues
of study.

\section{Background}

\label{sec:background}

\subsection{Task Scheduling Model}

The specific problem that we address here is the scheduling of a task
graph $G=\{V,E,w,c\}$ on a set of processors $P$. $G$ is a directed
acyclic graph wherein each node $n\in V$ represents a task, and each
edge $e_{ij}\in E$ represents a required communication from task
$n_{i}$ to task $n_{j}$. Figure \ref{fig:A-simple-task} shows an
example of a task graph. The computation cost of a task $n\in V$
is given by its positive weight $w(n)$, and the communication cost
of an edge $e_{ij}\in E$ is given by the non-negative weight $c(e_{ij})$.
The target parallel system for our schedule consists of a finite number
of homogeneous processors, represented by $P$. Each processor is
dedicated, meaning that no executing task may be preempted. We assume
a fully connected communication subsystem, such that each pair of
processors $p_{i},p_{j}\in P$ is connected by an identical communication
link. Communications are performed concurrently and without contention.
Local communication (from $p_{i}$ to $p_{i}$) is assumed to take
place in the memory shared by the tasks on the same processor and
is therefore assumed to have zero cost.

\begin{figure}
\begin{centering}
\includegraphics[width=0.3\columnwidth]{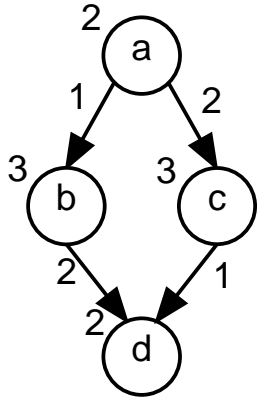}
\par\end{centering}
\caption{\label{fig:A-simple-task}A simple task graph.}
\end{figure}
Our aim is to produce a schedule $S=\{proc,t_{\mathrm{s}}\}$, where
$proc(n)$ allocates the task to a processor in $P$, and $t_{\mathrm{s}}(n)$
assigns it a start time on this processor. For a schedule to be valid,
it must fulfill two conditions for all tasks in $G$. The Processor
Constraint requires that only one task is executed by a processor
at any one time. The Precedence Constraint requires that a task $n$
may only be executed once all of its predecessors have finished execution,
and all required data has been communicated to $proc(n)$. Figure
\ref{fig:A-valid-schedule} illustrates a valid schedule for the task
graph in Figure \ref{fig:A-simple-task}. The goal of optimal task
scheduling is to find such a schedule $S$ for which the total execution
time or schedule length $sl(S)$ is the lowest possible.

\begin{figure}
\begin{centering}
\includegraphics[width=0.3\columnwidth]{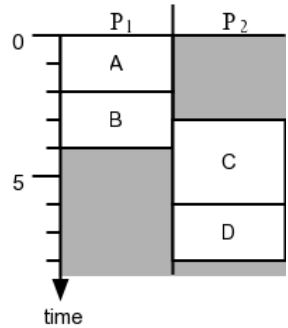}
\par\end{centering}
\caption{A valid schedule for the simple task graph of Fig.\label{fig:A-valid-schedule}.}
\end{figure}
It is useful to define the concept of node levels for a task graph
\citep{Sin07TSP}. For a task $n$, the top level $tl(n)$ is the
length of the longest path in the task graph that ends with $n$.
This does not include the weight of $n$, or any communication costs.
Similarly, the bottom level $bl(n)$ is the length of the longest
path beginning with $n$, excluding communication costs. The weight
of $n$ is included in $bl(n)$. The allocated top and bottom levels
$tl_{\alpha}(n)$ and $bl_{\alpha}(n)$ incorporate communication
costs, once the allocation of a parent and child task to different
processors confirms that the edge between them will be incurred.

Task graphs can be categorised into a number of different structures,
which often have distinct properties and associated difficulties when
solving. Figure \ref{fig:Simple-graphs} shows examples of some of
the simplest structures.

\begin{figure}
\hfill{}\subfloat[\label{fig:Fork-graph}Fork]{\begin{centering}
\includegraphics[width=4.5cm]{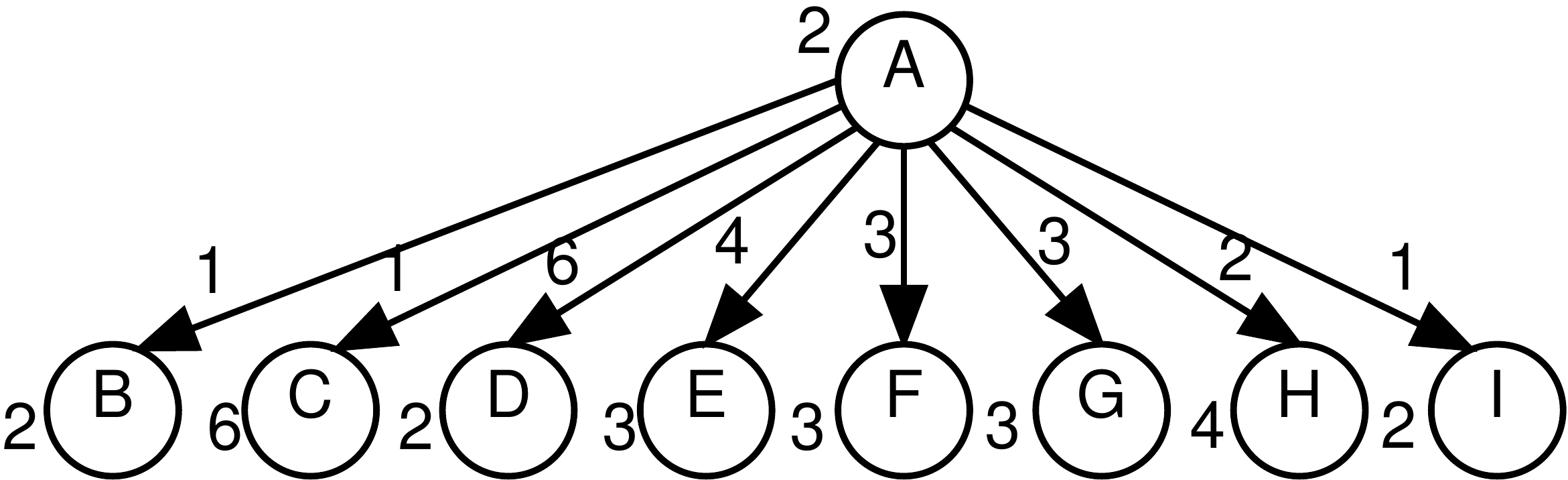}
\par\end{centering}
}\hfill{}\subfloat[\label{fig:Join-graph}Join]{\begin{centering}
\includegraphics[width=4cm]{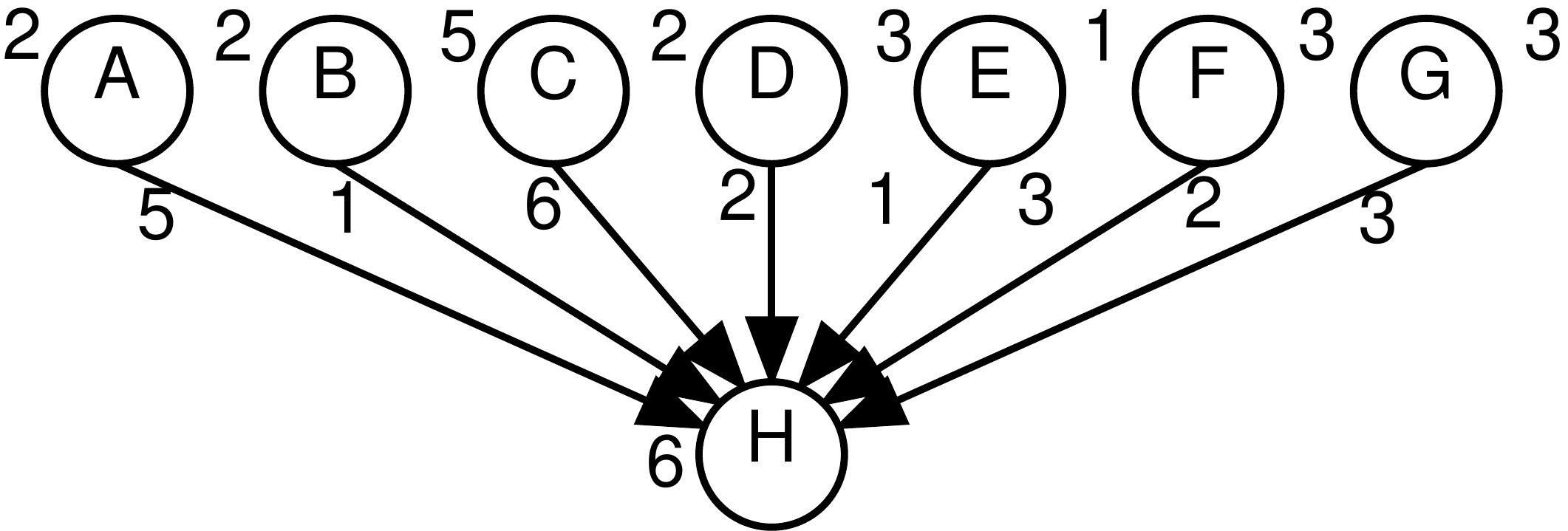}
\par\end{centering}
}\hfill{}\subfloat[\label{fig:Fork-join}Fork-join]{\begin{centering}
\includegraphics[width=2cm]{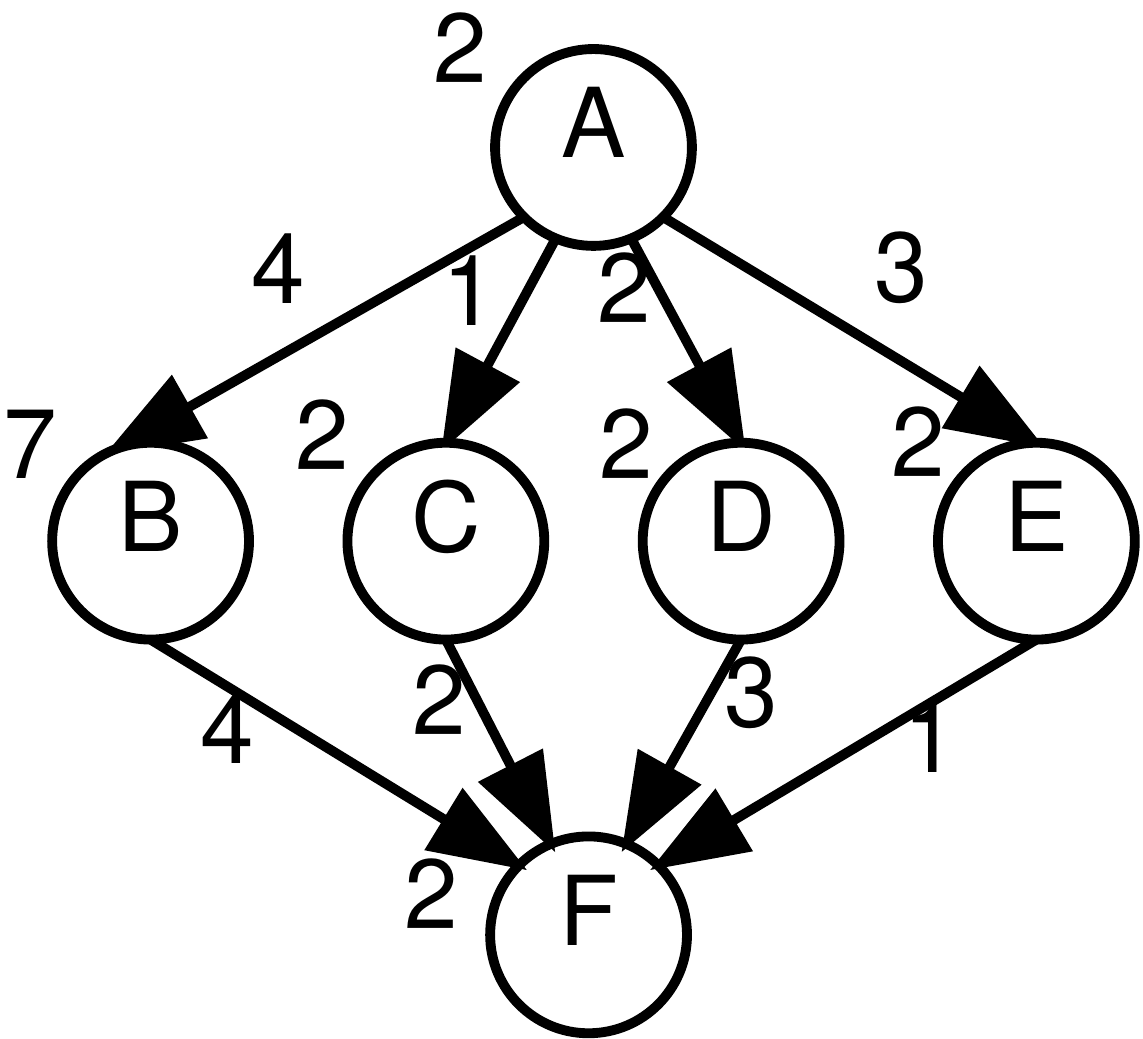}
\par\end{centering}
}\hfill{}

\caption{\label{fig:Simple-graphs}Example of simple task graph structures}
\end{figure}

\subsection{Branch-and-Bound}

The term branch-and-bound refers to a family of search algorithms
which are widely used for the solving of combinatorial optimisation
problems. They do this by implicitly enumerating all solutions to
a problem, simultaneously finding an optimal solution and proving
its optimality \citep{BundyWallen1984}. A search tree is constructed
in which each node (usually referred to as a state) represents a partial
solution to the problem. From the partial solution represented by
a state $s$, some set of operations is applied to produce new partial
solutions which are closer to a complete solution. In this way we
define the children of $s$, and thereby \textit{branch}. Each state
must also be \textit{bounded}: we evaluate each state $s$ using a
cost function $f$, such that $f(s)$ is a lower bound on the cost
of any solution that can be reached from $s$. Using these bounds,
we can guide our search away from unpromising partial solutions and
therefore remove large subsets of the potential solutions from the
need to be fully examined.

A{*} is a particularly popular variant of branch-and-bound which uses
a best-first search approach \citep{AICPub834:1968}. A{*} has the
interesting property that it is optimally efficient; using the same
cost function $f$, no search algorithm could find an optimal solution
while examining fewer states. To achieve this property, it is necessary
that the cost function $f$ provides an underestimate. That is, it
must be the case that $f(s)\leq f^{*}(s)$, where $f^{*}(s)$ is the
true lowest cost of a complete solution in the sub-tree rooted at
$s$. A cost function with this property is said to be \textit{admissable}.

\subsection{\label{subsec:Exhaustive-List-Scheduling}Exhaustive List Scheduling}

Previous branch-and-bound approaches to optimal task scheduling have
used a state-space model that is inspired by list scheduling algorithms
\citep{ShaSinAst2010}. States are partial schedules in which some
subset of the tasks in the problem instance have been assigned to
a processor and given a start time. At each branching step, successors
are created by putting every possible ready task (tasks for which
all parents are already scheduled) on every possible processor at
the earliest possible start time. In this way, the search space demonstrates
every possible sequence of decisions that a list scheduling algorithm
could make. This branch-and-bound strategy can therefore be described
as exhaustive list scheduling. Figure \ref{fig:ELS-branching} demonstrates
how four possible child states can be reached from a partial schedule
with two ready tasks and two processors.

\begin{figure}
\begin{centering}
\includegraphics[width=0.3\columnwidth]{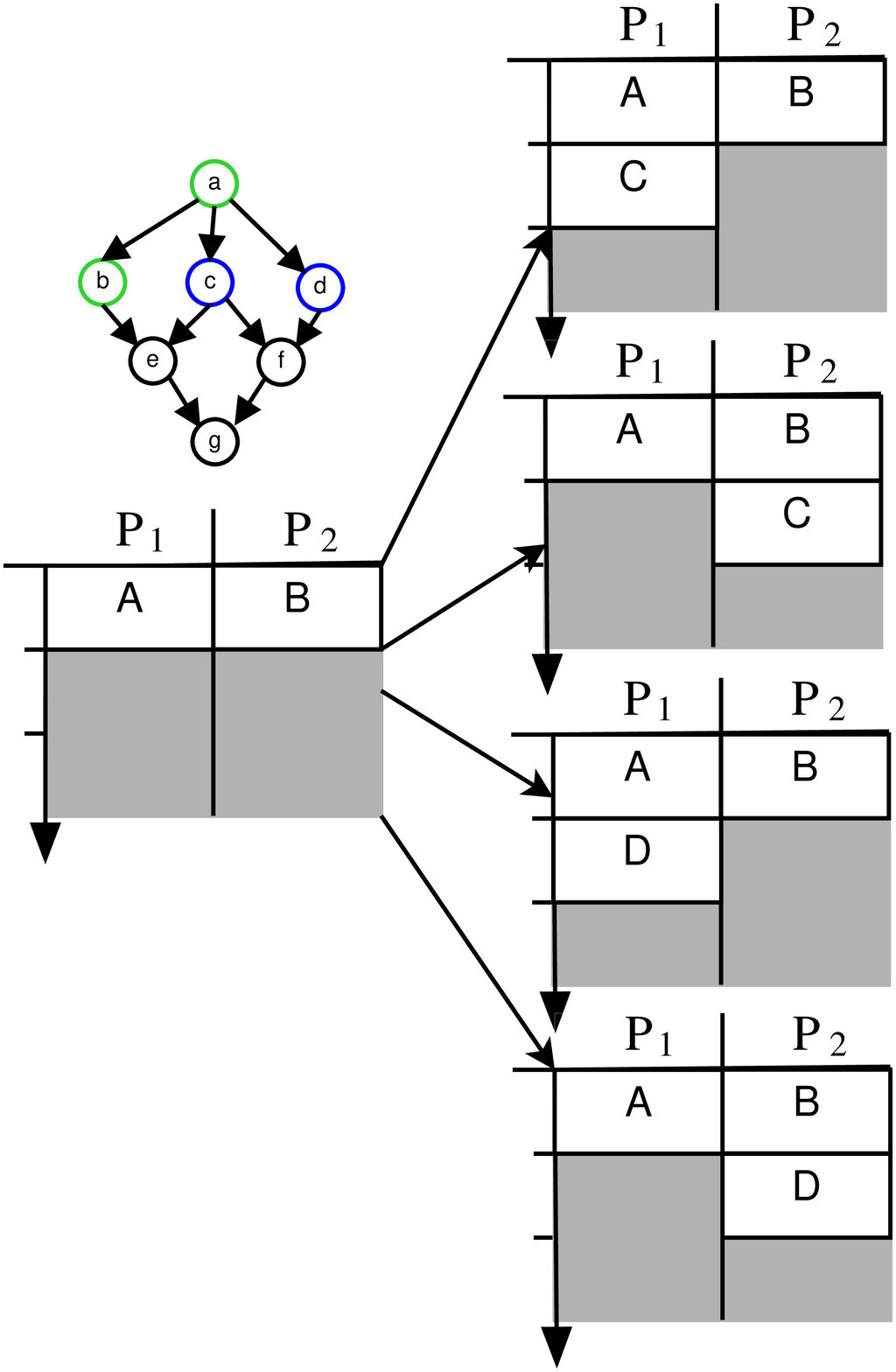}
\par\end{centering}
\caption{Branching in the ELS state space.\label{fig:ELS-branching}}
\end{figure}
Branch-and-bound works most efficiently when the sub-trees produced
when branching are entirely disjoint. Another way of stating this
is that there is only one possible path from the root of the tree
to any given state, and therefore there is only one way in which a
search can create this state. When this is not the case, a large amount
of work can be wasted: the same state could be expanded, and its sub-tree
subsequently explored, multiple times. Avoiding this requires doing
work to detect duplicate states, such as keeping a set of already
created states with which all new states must be compared. This process
increases the algorithm's need for both time and memory.

Unfortunately, the ELS strategy creates a lot of potential for duplicated
states \citep{Sinnen2014201}. This stems from two main sources:

\subsubsection*{Processor Permutation Duplicates}

Firstly, since the processors are homogeneous, any permutation of
the processors in a schedule represents an entirely equivalent schedule.
For example, Figure \ref{fig:Processor-permutation-duplicates} shows
two partial schedules which are identical aside from the labelling
of the processors. This means that for each truly unique complete
schedule, there will be $|P|!$ equivalent complete schedules in the
state space. This makes it very important to use some strategy of
processor normalisation when branching, such that these equivalent
states cannot be produced.

\begin{figure}
\begin{centering}
\includegraphics[width=0.5\columnwidth]{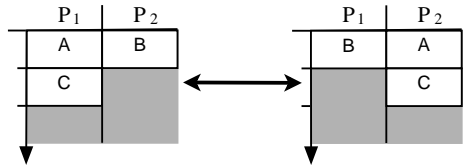}
\par\end{centering}
\caption{Processor permutation duplicates.\label{fig:Processor-permutation-duplicates}}
\end{figure}

\subsubsection*{Independent Decision Duplicates}

The other source of duplicate states is more difficult to deal with.
When tasks are independent of each other, the order in which they
are selected for scheduling can be changed without affecting the resulting
schedule. This means there is more than one path to the corresponding
state, and therefore a potential duplicate. Figure \ref{fig:Independent-decision-duplicates.}
demonstrates how multiple paths can lead to the same state. The only
way to avoid these duplicates is to enforce a particular sequence
onto these scheduling decisions. Under the ELS strategy, however,
no method is apparent in which this could be achieved while also allowing
all possible legitimate schedules to be produced.

\begin{figure}
\begin{centering}
\includegraphics[width=0.5\columnwidth]{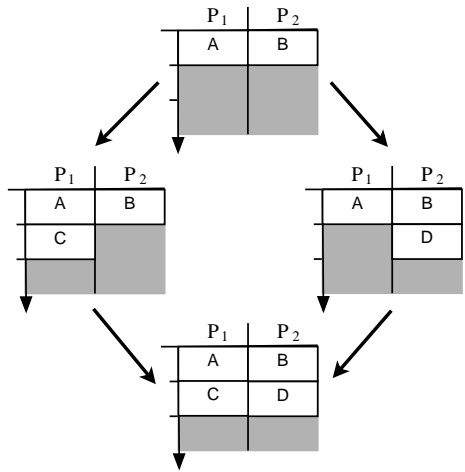}
\par\end{centering}
\caption{Independent decision duplicates.\label{fig:Independent-decision-duplicates.}}
\end{figure}

\section{Related Work}

Although the task scheduling model presented here can be applied to
the parallelisation of any arbitrary program, one specific area to
which task scheduling has been practically applied in recent years
is in the implementation of linear algebra solvers. Software packages
such as SuperMatrix\citep{chan2009foundations} and PLASMA\citep{kurzak2010scheduling}
represent linear algebra algorithms as DAGs, decomposing the steps
required into tasks, and use a variety of methods to schedule these
graphs. They do not, however, attempt to find optimal schedules, although
this could be helpful in some specific circumstances.

Due to the NP-hard nature of the task scheduling problem, the majority
of efforts have gone towards developing polynomial-time heuristic
solutions, producing schedules which are ``good enough'' for most
applications. The two major categories of approximation algorithms
are list scheduling and clustering \citep{Sin07TSP}. List scheduling
algorithms generally proceed in two phases: first, all of the tasks
are placed into an ordered list, according to some priority scheme.
Second, the tasks are removed from the list one by one, in order,
and scheduled. That is, they are assigned to a processor, and usually
given a start time as early as possible after the tasks previously
assigned to that processor have completed. Notable list scheduling
variants include MCP\citep{wu1990hypertool} and HLFET\citep{Kwok1998:btg}.
In a clustering algorithm, tasks are grouped together in clusters,
with the intention that if two tasks belong to the same cluster then
it is likely to be beneficial for them to be assigned to the same
processor. Technically, a clustering algorithm is only concerned with
suggesting a processor allocation for the tasks, with a subsequent
process similar to list scheduling usually required in order to produce
an actual schedule. Notable clustering variants include DCP\citep{kwok1996dynamic}
and DSC\citep{gerasoulis1992comparison}. While the ELS model resembles
a generalised version of the list scheduling approach (hence its name),
the new AO model more closely resembles a clustering approach.

The task scheduling problem is theoretically susceptible to many combinatorial
optimisation techniques. A{*}, the popular branch-and-bound search
algorithm, has been successfully applied to the optimal solving of
small problem instances through the ELS model \citep{ShaSinAst2010}
and earlier attempts \citep{KwoAhm2005OMT}. Another combinatorial
optimisation technique which has been applied to this task scheduling
problem is integer linear programming (ILP). This involves formulating
the problem instance as a linear program, a series of simultaneous
linear equations, where the variables are constrained to integer values.
A number of possible ILP formulations of the $P|prec,c_{ij}|C_{\mathrm{max}}$
problem have been proposed \citep{el2014new,mallach2016improved,Venugopalan2015:ilp},
with similarly promising results as branch-and-bound. Neither technique
has been shown to have a significant advantage over the other in terms
of the size of task scheduling problem that they can solve practically.
Most widely used ILP solvers are mature, highly optimised, proprietary
software packages. This optimisation means they are very likely to
have a built-in advantage in terms of speed when compared to a custom
implementation of state-space search. On the other hand, their complexity
and proprietary nature make them somewhat of a ``black box''. A
custom implementation makes it easier to gain potentially critical
insight into the behaviour of the solver. Additionally, it is often
easier to map domain-specific knowledge directly into a state-space
model than into an ILP formulation.

Many combinatorial optimisation problems can be naturally expressed
as permutation problems, meaning that the set of solutions consists
of all possible arrangements or orderings of some set of objects.
In the most famous of these, the travelling salesperson problem (TSP),
the goal is to plan a trip which visits a number of cities such that
the total distance travelled between the cities is minimised \citep{lin1965computer}.
Evidently, each possible ordering of the cities represents a valid
solution, and the optimal solution can be found by iterating over
these permutations. Other classic permutation problems include the
assignment problem \citep{munkres1957algorithms} and the MAX-SAT
problem \citep{borchers1998two}. Another class of combinatorial optimisation
problem can be thought of as a distribution or allocation problem,
in which a set of objects must be divided among a number of possible
groups. Examples include the graph colouring problem \citep{pardalos1998graph},
and the problem of task scheduling with independent tasks \citep{Sin07TSP}. 

The problem of task scheduling with communication delays is interesting
in that it can be considered as the composition of a distribution
problem with a permutation problem, as the tasks must both be optimally
divided among the processors and ordered optimally on each processor.
An example of a similar problem is that of minimizing part programs
for numerical control (NC) punch presses, which could be decomposed
into two distinct permutation problems (TSP and the assignment problem),
A proposed algorithm for this problem iterated between these two sub-problems,
using heuristic approaches to solve each \citep{walas1984algorithm}.
To the best of our knowledge, our proposed model is the first to attempt
optimal solving by branch-and-bound through separating two subproblems
into distinct phases, which are combined into one overall solution
space.

Pruning techniques are considered to be a fundamental aspect of branch-and-bound
search, as they have the potential to greatly limit the number of
states that need to be evaluated \citep{brusco2006branch}. Work on
pruning techniques for the ELS model demonstrated that they had a
dramatic impact on its performance \citep{Sinnen2014201}. It seems
unlikely that any state-space model could be competitive without the
development of effective pruning techniques.

\section{Duplicate-Free State-Space Model}

\label{sec:theory}

Both sources of duplicate states can be eliminated by adopting a new
state-space model (AO), in which the two dimensions of task scheduling
are dealt with separately. Rather than making all decisions about
a task's placement simultaneously, the search proceeds in two stages.
In the first stage, we decide for each task the processor to which
it will be assigned. We refer to this as the allocation phase. The
second stage of the search, beginning after all tasks are allocated,
decides the start time of each task. Given that each processor has
a known set of tasks allocated to it, this is equivalent to deciding
on an ordering for each set. Therefore, we refer to this as the ordering
phase. Once the allocation phase has determined the tasks' positions
in space, and the ordering phase has determined the tasks' positions
in time, a complete schedule is produced. Essentially, we divide the
problem of task scheduling into two distinct sub-problems, each of
which can be solved separately using distinct methods. However, while
we distinguish the two phases, we combine them into a single state-space,
making this a powerful approach.

\subsection{Allocation\label{subsec:Allocation}}

In the allocation phase, we wish to allocate each task to a processor.
Since the processors in our task scheduling problem are homogeneous,
the exact processor on which a task is placed is unimportant. What
matters is the way the tasks are grouped on the processors. The problem
of task allocation is therefore equivalent to the problem of producing
a partition of a set. A partition of a set $X$ is a set of non-overlapping
subsets of $X$, such that the union of the subsets is equal to $X$.
In other words, the set of all partitions of $X$ represents all possible
ways of grouping the elements of $X$. Applying this to our task scheduling
problem, we find all possible ways in which tasks could be grouped
on processors. In the allocation phase, we are therefore searching
for a partition of the set $V$ that can lead to an optimal schedule,
consisting of all tasks in our task graph. Figure \ref{fig:Partitions-of-tasks,}
shows different possible partitions of a set of tasks.

The search is conducted by constructing a series of partial partitions
of $V$. A partial partition $A$ of $V$ is defined as a partition
of a set $V^{\prime}$, $V^{\prime}\subseteq V$ \citep{Ronse13}.
At each level of the search we expand the subset $V^{\prime}$ by
adding one additional task $n\in V$, until $V^{\prime}=V$ and all
tasks are allocated. At each stage, the task $n$ selected can be
placed into any existing part $a\in A$ , or alternatively, a new
part can be added to $A$ containing only $n$. In our implementation
we build an allocation by adding tasks in a topological order, but
any pre-defined order suffices. As we are allocating tasks to a finite
number of processors $|P|$, we simply limit the number of parts allowed
in a partial partition to the same number. This has no effect other
than to reduce the search space by disregarding partitions consisting
of a larger number of sets. A partial allocation $A$ will have $|A|$
children again of size $|A|$ and one child of size $|A|+1$. Therefore,
as the number of parts in a partial allocation is non-decreasing as
we move deeper in the search tree, disregarding partial allocations
such that $|A|>|P|$ cannot prevent valid allocations where $|A|\leq|P|$
from being discovered. Figure \ref{fig:AO-branching} illustrates
how child states are derived from a partial partition: the new task
$D$ can be added to either of the existing parts. If scheduling with
three or more processors, it could be placed in a new part by itself,
but for this example we limit the allocation to two processors.

\begin{figure}
\begin{centering}
\includegraphics[width=0.5\columnwidth]{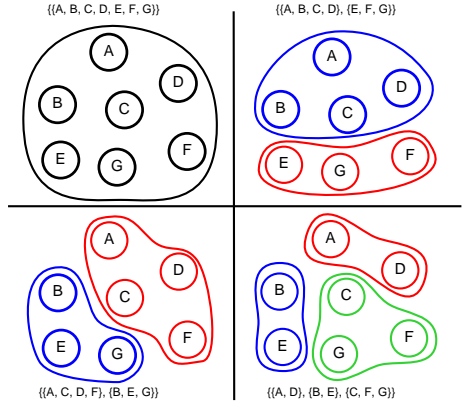}
\par\end{centering}
\caption{Partitions of tasks, equivalent to processor allocations.\label{fig:Partitions-of-tasks,}}
\end{figure}
\begin{lem}
\label{lem:allocation}\textbf{The allocation phase of the AO model
can produce all possible partitioning of tasks and there is only one
unique sequence that produces each possible paritition.}
\end{lem}
\begin{proof}
We show how any given allocation $A_{\Omega}$, which is a complete
partition of $V$, is constructed with the proposed allocation procedure
and that there is only one possible choice at each step. We begin
with an empty partial allocation $A=\{\}$, and are presented with
the tasks in $V$ in a fixed order $n_{1}$, $n_{2}$, ... $n_{|V|}$.
We must always begin by placing $n_{1}$ into a new part, so that
$A=\{\{n_{1}\}\}$. Now we must place $n_{2}$. If $n_{1}$ and $n_{2}$
belong to the same part in $A_{\Omega}$, they must also be placed
in the same part in $A$, and so we must have $A=\{\{n_{1},n_{2}\}\}.$
Conversely, if $n_{1}$ and $n_{2}$ belong to different parts in
$A_{\Omega}$, the same must be true in $A$, and we make $A=\{\{n_{1}\},\{n_{2}\}\}.$
For each subsequent task $n_{i}$, if $n_{i}$ in $A_{\Omega}$ belongs
to the same part as any of tasks $n_{1}$ to $n_{i-1}$, we must place
it in the same part in $A$. If $n_{i}$ does not share a part in
$A_{\Omega}$ with any task $n_{1}$ to $n_{i-1}$, it must be placed
in a new part by itself. At each step, there is exactly one possible
move that can be taken in order to keep $A$ consistent with $A_{\Omega}$.
Since there is always at least one move, it is possible to produce
any partition of $V$ in this fashion, and because there is always
at most one move, there is only one path that will produce a given
partition. Since each distinct allocation can only be produced with
one unique sequence of moves, duplicate allocations are not possible.
\end{proof}
By Lemma \ref{lem:allocation}, we have removed the first source of
duplicates: there is no possibility of producing allocations that
differ from each other only by the permutation of processors.

In a naive approach to allocation, in which we simply assign each
task in $V$ to an arbitrary processor in $P$ without considering
their homogeneity, there are $|P|^{|V|}$ possible outcomes. The number
of possible complete allocations in this state space is given by the
formula $\sum_{k=1}^{|P|}\begin{Bmatrix}|V|\\
k
\end{Bmatrix}$, where $\begin{Bmatrix}n\\
k
\end{Bmatrix}$ represents the number of distinct ways to divide a set of size $n$
into $k$ subsets (this being known as a Stirling number of the second
kind) \citep{stephen1975some}. In the worst case, where the number
of available processors is equal to the number of tasks, the formula
gives us what is known as a Bell number: the total number of possible
partitions of a set with size $n$. Bell numbers are known to be asymptotically
bounded such that $B_{n}<(\frac{0.792n}{ln(n+1)})^{n}$\citep{berend2010improved}.
This bound demonstrates that, although the number of allocations still
grows exponentially, it is an exponential function of significantly
lower order than the naive $|P|^{|V|}$approach.

\begin{figure}
\begin{centering}
\includegraphics[width=0.75\columnwidth]{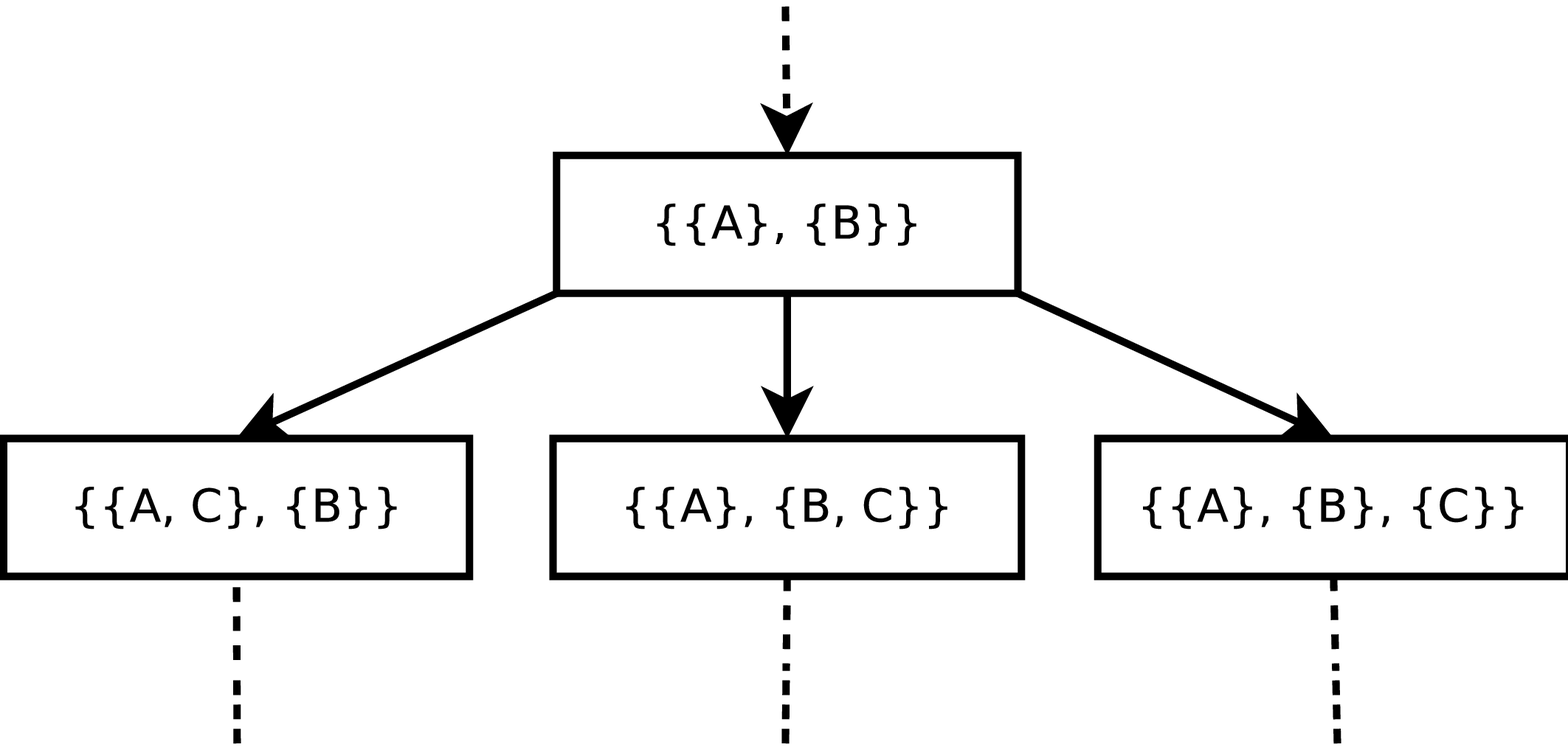}
\par\end{centering}
\caption{Branching in the allocation state-space with a maximum of two processors.\label{fig:AO-branching}}
\end{figure}
\begin{algorithm}[h]
\begin{algorithmic}
\Require $A_s$ is a partial partition of $V$
\Function{expand-alloc-state}{$s$, $V$, $|P|$}
	\State $unallocated \gets V \setminus \bigcup\nolimits_{a \in A_s} a$
	\State $nextTask \gets min_{n \in unallocated} topo\mhyphen order(n, G)$
	\State $children_s \gets \emptyset$
	\ForAll{$a \in A_s$}
		\State $s_child \gets$ a copy of $s$
		\State $a_child \gets$ the set in $A_child$ which is equal to $a$
		\State $a_child \gets a_child \cup \{nextTask\}$
		\State $children_s \gets children_s \cup \{s_child\}$
	\EndFor
	\If{$\lvert A_s \rvert < |P|$}
		\State $s_child \gets$ a copy of $s$
		\State $a_new \gets \{nextTask\}$
		\State $A_child \gets A_child \cup \{a_new\}$
		\State $children_s \gets children_s \cup \{s_child\}$
	\EndIf
\Return $children_s$
\EndFunction
\end{algorithmic}

\caption{Defining children of an allocation state.}
\end{algorithm}

\subsubsection{\label{subsec:Allocation-Cost-Function}Allocation Cost Function}

For branch-and-bound search using our AO state-space model to be effective,
we need an admissable heuristic to determine our cost function $f$,
such that $f(s)$ gives a lower bound for the minimum length of any
schedule resulting from the partial partition $A$ at state $s$.
The efficiency of a branch-and-bound search can be defined by the
number of states in the state-space that need to be examined before
a provably optimal solution can be found \citep{ibaraki1977computational}.
To prove that a solution is optimal, all states in the state-space
with a lower $f$-value must be examined and shown to not be solutions
themselves. Such states are said to be 'critical' \citep{talbi2006parallel}.
A tighter bound can make the search more efficient, while a looser
bound may make the search less efficient. This is because a tighter
bound is likely to increase the $f$-value of some states, possibly
making them no longer critical, while a looser bound can do the opposite.

In the case of allocation, there are two crucial types of information
we can obtain from a partial partition $A$, which allow us to determine
a lower bound for the length of a resulting schedule. The first, and
simplest, is how well the computational load is balanced between the
different groupings of tasks. In an ideal case, with no gaps in execution,
a processor assigned a grouping $a\in A$ will still require time
equal to the sum of the computational weights $w(n)$ of the tasks
$n\in a$, i.e. $\sum_{n\in a}w(n)$, in order to finish. The overall
schedule length is therefore bounded by the total weight of the most
heavily loaded grouping in $A$. We improve this bound further by
considering the time at which execution can start on each processor.
The earliest possible start time for a task $n$ is given by its allocated
top level $tl_{\alpha}(n)$. The earliest that a processor can begin
executing tasks can therefore be determined by finding the minimum
allocated top level among its assigned tasks $n\in a$. Similarly,
the minimum bottom level (disregarding the weight of the first task
in the level path) among the grouping tells us the minimum amount
of time required to reach the end of the schedule once all tasks on
this processor have finished execution. We therefore add the minimum
allocated top level and minimum bottom level (disregarding (disregarding
the weight of first task in the level path) to each grouping's computational
load to obtain our first bound.

\begin{align}
f_{l\mathrm{oad}}(s)=\mathrm{max}{}_{a\in A} & \biggl\{ min_{n\in a}tl_{\alpha}(n)+\sum_{n\in a}w(n)\nonumber \\
 & +min_{n\in a}(bl_{\alpha}(n)-w(n))\biggr\}\label{eq:fload-bl}
\end{align}

The second bound derives from our knowledge of which communication
costs must be incurred, and is obtained from the length of the allocated
critical path of the task graph; that is, the longest path through
the task graph given the particular set of allocations. It is usual
when finding a critical path for a task graph to ignore the edge weights,
as in an ideal case no communication costs would need to be incurred.
However, if two tasks $i,j\in V$ have been assigned to different
groupings in $A$, and an edge $e_{ij}$ exists between them, we now
know that the communication represented by that edge must take place.
Such edges can therefore be included when determining the critical
path, and may increase its length. As all computations and communications
in the allocated critical path must be completed in sequence, the
overall schedule must be at least as long, and this gives us our second
bound. Figure \ref{fig:allocated-critical-heuristic} shows an example
of an allocated critical path. In practise, it is helpful to consider
the allocated top and bottom levels of the tasks when determining
the allocated critical path, since we use those values for several
other calculations. The length of the longest path in a task graph
which includes a task $n$ can be found by adding the top and bottom
levels of $n$. It follows that the length of the allocated critical
path is the greatest value for $tl_{\alpha}(n)+bl_{\alpha}(n)$ among
all $n\in V^{\prime}$. 

\begin{equation}
f_{\mathrm{acp}}(s)=\mathrm{max}_{n\in V^{\prime}}\left\{ tl_{\mathrm{\alpha}}(n)+bl_{\alpha}(n)\right\} \label{eq:facp}
\end{equation}
\begin{figure}
\begin{centering}
\includegraphics[width=0.3\columnwidth]{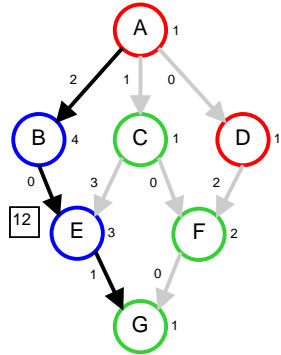}
\par\end{centering}
\caption{The allocated critical path heuristic.\label{fig:allocated-critical-heuristic}}
\end{figure}
Since we want the tightest bound possible, the maximum of these two
bounds is taken as the final $f$-value. 

\begin{equation}
f_{\mathrm{alloc}}(s)=\mathrm{max}\{f_{\mathrm{load}}(s),f_{\mathrm{acp}}(s)\}
\end{equation}

By their nature, these two bounds oppose each other; lowering one
is likely to increase the other. The shortest possible allocated critical
path can be trivially obtained simply by allocating all tasks to the
same processor, but this will cause the total computational weight
of that processor to be the maximum possible. Likewise, the lowest
possible computational weight on a single processor can be achieved
simply by allocating each task to a different processor, but this
means that all communication costs will be incurred and therefore
the allocated critical path will be the longest possible. Combining
these two bounds guides the search to find the best possible compromise
between computational load-balancing and the elimination of communication
costs.

\subsection{\label{subsec:Ordering}Ordering}

In the ordering phase, we begin with a complete allocation, and our
aim is to produce a complete schedule $S$. After giving an arbitrary
ordering to both the sets in $A$ and the processors in $P$, we can
define the processor allocation in $S$ such that $n\in a_{i}\implies proc(n)=p_{i}$.
 The remaining step is to determine the optimal start time for each
task. Given a particular ordering of the tasks $n\in p_{i}$, the
best start time for each task is trivial to obtain, as it is simply
the earliest it is possible for that task to start, considering the
availability of the processor and the precedence constraints induced
by the incoming edges. To complete our schedule we therefore only
need to determine an ordering for each set of tasks $a_{i}\in A$.
Our search could proceed by enumerating all possible permutations
of the tasks within their processors. However, it is likely that many
of the possible permutations do not describe a valid schedule. This
will occur if any task is placed in order after one of its descendants
(or before one of its ancestors).

In order to produce only valid orderings, an approach inspired by
list scheduling is taken. In this variant, however, each processor
$p_{i}$ is considered separately, with a local ready list $R(p_{i})$.
Initially, a task $n\in p_{i}$ is said to be locally ready if it
has no predecessors also on $p_{i}$. At each step we can select a
task $n\in R(p_{i})$ and place it next in order on $p_{i}$. Those
tasks which have been selected and placed in order are called \emph{ordered,
}while those which have not are called \emph{unordered. }A simple
definition of the local ready list states that a task $n\in p_{i}$
belongs to $R(p_{i})$ if it has no unordered predecessors also on
$p_{i}$. Unfortunately, this formulation allows invalid states to
be reached by the search, as seemingly valid local orders may combine
to produce a schedule with an invalid global ordering. The simple
definition still guarantees that all valid schedules can be produced,
and invalid branches of the state-space can easily be removed from
consideration: the f-value of certain invalid states is undefined,
and in calculcation will increase indefinitely. However, an indeterminate
amount of work may be wasted in exploring these invalid states. A
correct formulation of the local ready list requires that any task
$n_{i}$ which is ordered on $p_{i}$ must be subsequently considered
to be a predecessor of any task $n_{j}$ on $p_{i}$ which is still
unordered. This more complex ready condition is explained in detail
in Section \ref{sec:Avoiding-Invalid-Stat}. 

After a task $n$ has been ordered, each of its descendants on $p_{i}$
must be checked to see if  the ready condition has now been met, in
which case they will be added to $R(p_{i})$. Following this process
to the end, we can produce any possible valid ordering of the tasks
on $p_{i}$. Figure \ref{fig:BranchingA-single-state} illustrates
a single state in an ordering state-space, and shows the options for
branching which are available. In this example, several tasks have
already been ordered on processor $P1$, and it has again been selected
for consideration. Since they have no unordered predecessors also
on $P1$, tasks $e$ and $f$ are both ready to be ordered. There
are therefore two possible children of this state, one in which each
of these ready tasks is placed next in order on $P1$. Note that $f$
is ready to be ordered despite the fact that its parent, $d$, has
not been ordered. Since $d$ is allocated to a different processor,
it is not considered when determining the readiness of $f$. Neither
is its own parent, $a$, which would necessarily have been the very
first task scheduled under the ELS model. 

\begin{figure}
\begin{centering}
\includegraphics[width=0.4\columnwidth]{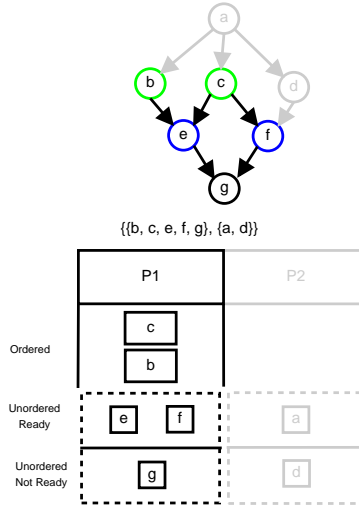}
\par\end{centering}
\caption{\label{fig:BranchingA-single-state}A single state in the ordering
state space.}
\end{figure}
Producing a full schedule requires that this process be completed
for all processors in $P$. At each level of the search, we can select
a processor $p_{i}\in P$ and order one of its tasks. The order in
which processors are selected can be decided arbitrarily; however,
in order to avoid duplication, it must be fixed by some scheme such
that the processor selected can be determined solely by the depth
of the current state. The simplest method to achieve this is to proceed
through the processors in order: first order all the tasks on $p_{1}$,
then all the tasks on $p_{2}$, and so on to $p_{n}$. Another method
is to alternate between the processors in a round-robin fashion. Unlike
in exhaustive list scheduling, tasks are not guaranteed to be placed
into the schedule in topological order. When a task is ordered, its
predecessors on other processors may still be unordered, and therefore
their start times may not be known. During the ordering process, therefore,
a task $n$ may only be given an \emph{estimated earliest start time
}$eest(n)$. For all unordered tasks, $eest(n)=tl_{\mathrm{\alpha}}(n)$,
its allocated top level. For ordered tasks, we first define $prev(n)$
as the task ordered immediately before $n$ on the same processor
$proc(n)$. We also define the estimated data ready time $edrt(n_{j})=max_{n_{i}\in parents(n_{j})}\left\{ eest(n_{i})+w(n_{i})+c(e_{ij})\right\} $.
Where $prev(n)$ does not exist, $eest(n)=edrt(n)$. Otherwise, $eest(n)=max(eest(prev(n))+w(prev(n)),edrt(n))$.
These are the same as the conditions for determining the earliest
start time of a task in ELS whose ancestors have already been scheduled,
but replacing the fixed, known start times of these ancestors with
estimated earliest start times in each instance.

In our implementation, the changes in estimated earliest start times
caused by the ordering of each new task $n_{\Delta}$ are propagated
recursively. First, $eest(n_{\Delta})$ is calculated based on the
EEST of the parents of $n_{\Delta}$, and the EEST of the task immediately
preceding it on its processor $proc(n_{\Delta})$. The algorithm then
proceeds to update the EEST of all ordered tasks whose start times
depend on $n_{\Delta}$ - any of its children which, being allocated
to a different processor, may have already been ordered. Once the
EEST of all these tasks has been recalculated, we continue propagating
to any previously ordered tasks which depend on them, and so on. In
this case, dependent tasks include not only children, but also tasks
which may be scheduled immediately after them on their respective
processor. The red arrows in figure \ref{fig:Updating-estimated-earliest}
show how the EEST updates are propogated after the ordering of task
$a$, both through the original communication dependencies in the
task graph and the new dependencies determined by the previously decided
ordering.

\begin{figure}
\begin{centering}
\includegraphics[width=0.4\columnwidth]{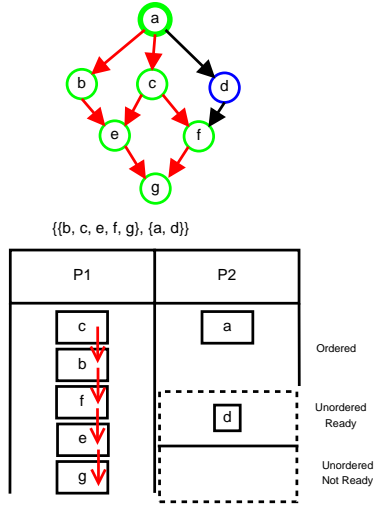}
\par\end{centering}
\caption{\label{fig:Updating-estimated-earliest}Updating estimated earliest
start times after ordering a new task.}
\end{figure}
In this way, we have solved the problem of duplicates arising from
making the same decisions in a different order. By allocating each
task to a processor ahead of time, and enforcing a strict order on
the processors, it is no longer possible for these situations to arise.
Where before we might have placed task $n_{2}$ on $p_{2}$ and then
task $n_{1}$ on $p_{1}$, we now must always place task $n_{1}$
on $p_{1}$ and then task $n_{2}$ on $p_{2}$. 
\begin{lem}
\textbf{The ordering phase of the AO model can produce all possible
valid orderings of tasks and there is only one unique sequence that
produces each possible ordering. }
\end{lem}
\begin{proof}
This proof is similar to that of Lemma~\ref{lem:allocation}. We
show how any given complete valid schedule $S_{\Omega}$, which implies
the complete allocation $A_{x}$, is constructed with the proposed
ordering procedure and that there is only one possible choice at each
step. Consider the sequence of moves required to replicate $S_{\Omega}$.
We begin with an empty schedule $S$, and then select processors for
consideration in a fixed and deterministic order (e.g. round robin).
Say that we select $p_{i}$. In schedule $S_{\Omega}$, there is a
task $n_{x}$ which is next in order on $p_{i}$. Since $S_{\Omega}$
is a valid schedule consistent with $A_{x}$, $n_{x}$ must belong
to the current ready queue for $p_{i}$. Since $n_{x}$ is next in
order in $S_{\Omega}$, it must be selected to be ordered next in
$S$. At each step, there is exactly one possible move that can be
taken in order to keep $S$ consistent with $S_{\Omega}$. Since there
is always at least one move, it is possible to produce any schedule
consistent with $A_{x}$ in this fashion, and because there is always
at most one move, there is only one path that will produce a given
schedule. Since each distinct schedule can only be produced with one
unique sequence of moves, duplicate schedules are not possible.
\end{proof}
Please note that this lemma is \emph{not} ruling out the creation
of invalid orderings. How they are avoided is discussed in \ref{sec:Avoiding-Invalid-Stat}.

\subsubsection*{Ordering Cost Function}

The heuristic for determining $f$-values in the ordering stage follows
a similar pattern to that for allocation. The difference lies in the
fact that independent tasks allocated to the same processor can now
delay one another. During the allocation phase, our critical path
heuristic assumes that every task begins as early as it theoretically
could on the processor it is assigned to, only based on the allocated
top-level. Another way of looking at this is that we assume that every
task will be first in order (and that this is also true for all ancestors).
Clearly this is not the case, and as we decide the actual order of
the tasks on a processor $p$, the tasks which are placed earlier
in the order are likely to push back the start times of the tasks
placed later in the order, as they must wait to be executed. Communications
from other processors can also introduce idle times, during which
the processor does nothing as the data required for the task next
in order is not yet ready. The $eest$ of a task takes both of these
factors into account. For each state $s$, the current estimated finish
time of a processor $p_{i}$ is the latest estimated finish time of
any task $n\in V:proc(n)=p_{i}$ which has so far been ordered. This
estimated finish time must include both the full computation time
of each task already ordered on $p_{i}$, as well as any idle time
incurred between tasks.

With this in mind, we define our two bounds like so: first, the latest
estimated start time of any task already ordered, plus the allocated
bottom level of that task. We refer to this as the partially scheduled
critical path, as it corresponds to the allocated critical path through
our task graph, but with the addition of the now known idle times
and intra-processor communication delays.

\begin{equation}
f_{\mathrm{scp}}(s)=\mathrm{max}_{n\in ordered(s)}\left\{ eest(n)+bl_{\mathrm{\alpha}}(n)\right\} 
\end{equation}

Second, the latest finish time of any processor in the partial schedule,
plus the total computational weight of all tasks allocated to that
processor which are not yet scheduled.

\begin{equation}
f_{\mathrm{ordered-load}}(s)=\mathrm{max}{}_{p\in P}\left\{ t_{\mathrm{f}}(p)+\sum_{n\in p\cap unordered(s)}w(n)\right\} 
\end{equation}

Again, this corresponds to the total computational load on a processor
with the addition of now known idle times and intra-processor delays.
To obtain the tightest possible bound, the maximum of these bounds
is taken as the final $f$-value.

\begin{equation}
f_{\mathrm{order}}(s)=\mathrm{max}\{f_{\mathrm{scp}}(s),f_{\mathrm{ordered-load}}(s)\}
\end{equation}

\subsection{Combined State-Space}

Solving a task scheduling problem instance requires both the allocation
and ordering sub-problems to be solved in conjunction. To produce
a combined state-space, we begin with the allocation search tree,
$S_{\mathrm{A}}$. The leaves of this tree represent every possible
distinct allocation of tasks in $G$ to processors in $P$. Say that
leaf $l_{i}$ represents allocation $A_{i}$. We produce the ordering
search tree $S_{\mathrm{O}_{i}}$ using $A_{i}$. The leaves of $S_{\mathrm{O}_{i}}$
represent every distinct complete schedule of $G$ which is consistent
with $A_{i}$. If we take each leaf $l_{i}$ in $S_{A}$ and set the
root of tree $S_{\mathrm{O}_{i}}$ as its child, the result is the
combined tree $S_{\mathrm{AO}}$, the leaves of which represent every
distinct complete schedule of $G$ on the processors in $P$. Figure
\ref{fig:combined-state-space} demonstrates how all of the ordering
sub-trees in $S_{AO}$ sprout from the leaves of the allocation tree
above them.

\begin{figure}[h]
\begin{centering}
\includegraphics[width=0.5\columnwidth]{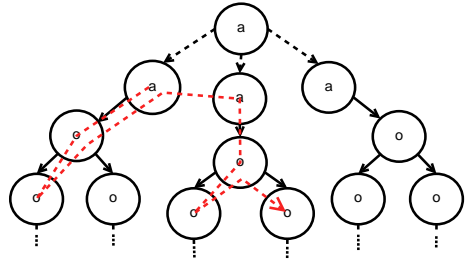}
\par\end{centering}
\caption{\label{fig:combined-state-space}A possible search path through the
combined state space.}
\end{figure}
A branch-and-bound search conducted on this state-space will begin
by searching the allocation state-space. Each allocation state representing
a complete allocation has one child state, which is an initial ordering
state with this allocation. When considering the allocation sub-problem
in isolation, we define the optimal allocation as that which has the
smallest possible lower bound on the length of a schedule resulting
from it. Unfortunately, these lower bounds cannot be tight and therefore
it is not guaranteed that the allocation with the smallest lower bound
will actually produce the shortest possible schedule. This means generally
that in the combined state-space, a number of complete allocations
are investigated by the search and have their possible orderings evaluated,
as represented by the red arrow in figure \ref{fig:combined-state-space}.
The tighter the bound which can be calculated, the more quickly the
search is likely to be guided toward a truly optimal allocation.

\section{Avoiding Invalid States\label{sec:Avoiding-Invalid-Stat}}

The simple definition of a ready list as done in Section~\ref{subsec:Ordering}
enforces valid local orders for all processors, but in same cases
their combination can be an invalid global ordering. To explain why
this happens, we model a partial schedule as a graph showing all of
the dependencies between tasks. For a task graph $G$, a partial schedule
$S^{\prime}$can be represented by augmenting $G$ to produce a partial
schedule graph $G_{S^{\prime}}.$ We begin with the graph $G$. Say
that in $S^{\prime}$, a task $n_{1}$ is ordered on processor $p_{1}$,
and a task $n_{2}$ is also ordered on $p_{1}$, but later in the
sequence. We check for the edge $e_{12}$ in the task graph. If $e_{12}\notin E$,
we add $e_{12}$to $E$. This new edge represents that, according
to the ordering defined by our partial schedule, $n_{2}$must begin
after $n_{1}$. This can be considered as a new type of dependency,
which we call an ordering dependency, as opposed to the original communication
dependencies in $G$. Once edges have been added corresponding to
all ordered tasks in $S^{\prime}$, we have our graph $G_{S^{\prime}}$.
The presence of a cycle in this graph indicates that the ordering
is invalid, as a cycle of dependencies is unsatisfiable. Since the
graph $G$ is acyclic, and if $n_{i}$ is an ancestor of $n_{j}$
in $G$ then the ordering edge $e_{ji}$ cannot be created, it is
necessary that any cycle in $G_{S^{\prime}}$ will contain at least
two ordering edges. Figure \ref{fig:cycle-example} shows a minimal
example of such a cycle, with ordering edges marked by dashed lines.

\begin{figure}
\begin{centering}
\subfloat[Task graph and allocation.]{\begin{centering}
\includegraphics[width=0.4\textwidth]{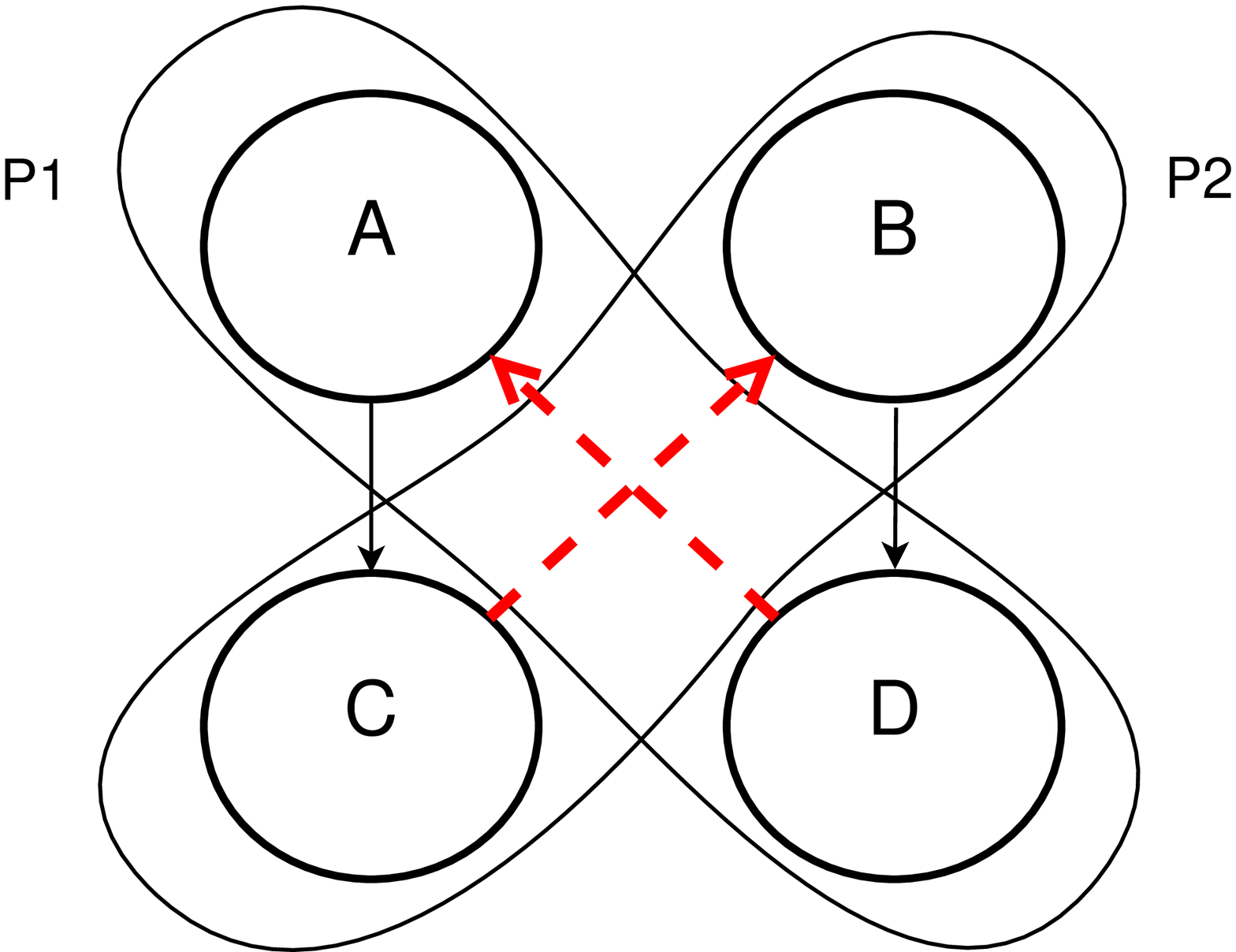}
\par\end{centering}
}\subfloat[Ordering which produces a cycle.]{\begin{centering}
\includegraphics[width=0.4\textwidth]{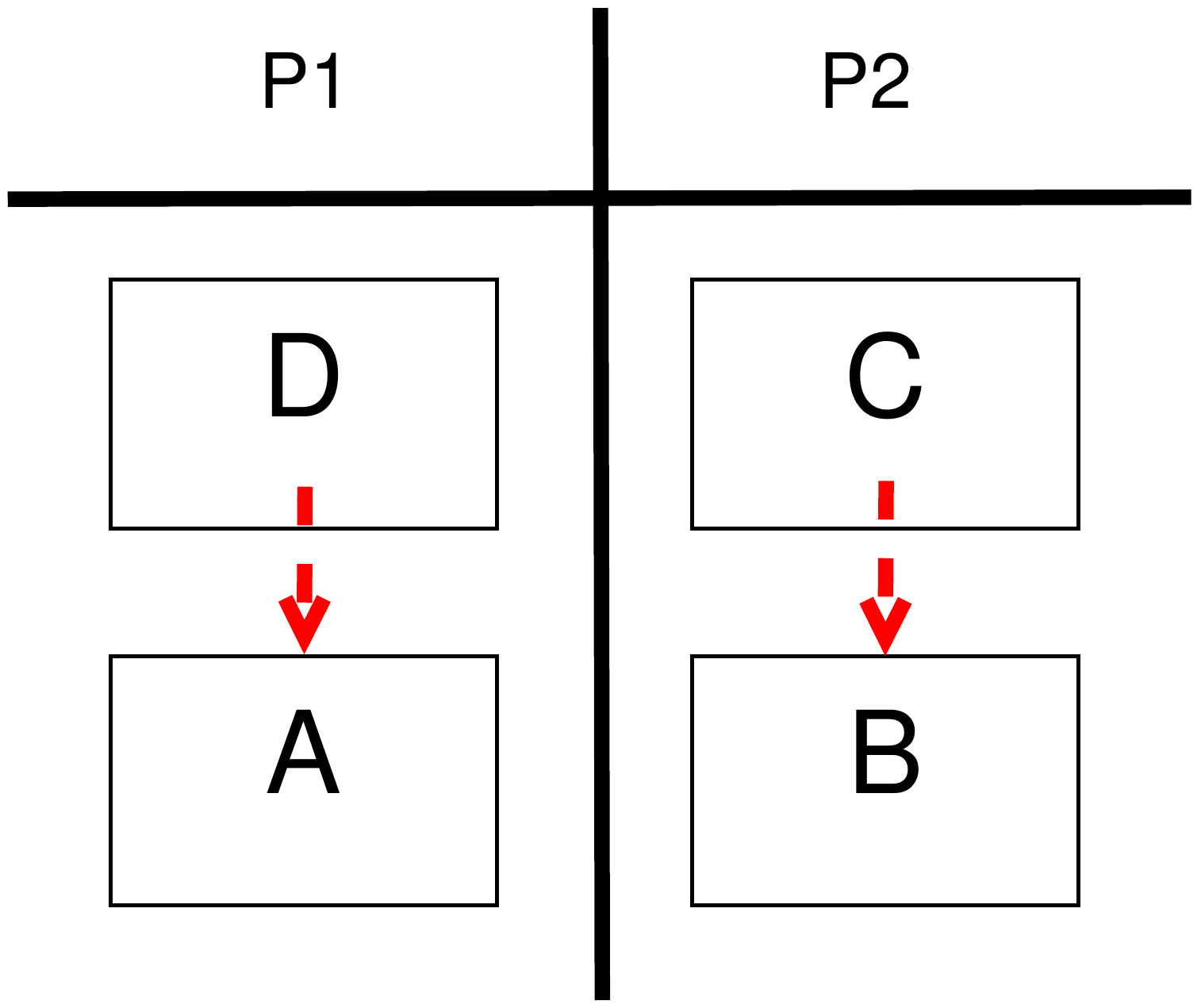}
\par\end{centering}
}
\par\end{centering}
\centering{}\caption{A minimal example of a cycle created by ordering edges.\label{fig:cycle-example}}
\end{figure}
In our preliminary work in \citep{Orr2015:dfs}, such states were
removed from consideration during the search as their cyclic nature
made their $f$-values increase infinitely during calculation, until
they passed an upper bound for the schedule length and it was clear
they could be ignored. However, it is possible for a state to exist
in which no cycle yet exists, but for which it is inevitable that
a cycle will be created as the ordering process continues. Say, for
example, that the introduction of edge $e_{ij}$ would create a cycle
in $G_{S^{\prime}}$, but in $S^{\prime}$ the task $n_{i}$ has already
been ordered while $n_{j}$ has not. In order for the schedule to
be completed, $n_{j}$ must eventually be ordered, at which point
a cycle will be formed. Here, $S^{\prime}$represents an entire subtree
of states from which no valid schedule can be reached. None of these
states can be selected as the optimal solution, so this does not present
a threat to the accuracy of the search process. However, it does represent
a potentially substantial amount of wasted work performed by the search
algorithm. Ideally the formulation of the AO model would be such that
it allows the creation of any valid solution, and \emph{only }valid
solutions.

The key to avoiding this unnecessary work is the observation that,
given that $n_{i}$ has been ordered and $n_{j}$ has not, it is inevitable
that $n_{j}$must eventually be ordered later than $n_{i}$. Therefore,
for all descendents of this partial schedule in which $n_{j}$ is
ordered, the ordering edge $e_{ij}$ must be in $G_{S^{\prime}}$.
We can therefore define a more useful augmented task graph, $G_{S^{\prime}}^{*}$,
which 'looks ahead' to determine cycles that must inevitably occur.
In this graph, the ordering edge $e_{ij}$ exists if $proc(n_{i})=proc(n_{j})$
and either $n_{j}$ is ordered later than $n_{i}$, or $n_{i}$ has
been ordered and $n_{j}$ has not. 

To avoid these cycles, we propose a modification to the condition
we use to determine if a task is free to be ordered. The new condition
is this: a task $n_{i}$ on processor $p_{i}$ is free to be ordered
if it has no ancestors in graph $G_{S^{\prime}}^{*}$ which are also
on $p_{i}$ and have not already been ordered in $S^{\prime}$. In
the original formulation of AO, this condition used only the graph
$G$. However, the ordering edges specific to the partial solution
$S^{\prime}$ must be considered equally with the communication edges
that are common to all partial solutions. The creation of a cycle
in $G_{S^{\prime}}^{*}$ requires that an ordering edge is introduced
from a task $n_{i}$ to task $n_{j}$, where $n_{i}$ was already
reachable from $n_{j}$ using at least one ordering edge. By definition,
this means that $n_{j}$ is the ancestor of $n_{i}$ in $G_{S^{\prime}}^{*}$.
Therefore, according to the new condition, $n_{i}$ cannot be considered
free until $n_{j}$ has been ordered, meaning that the edge $e_{ij}$
can never be introduced and the cycle can never be formed. By treating
the ordering dependencies created during the ordering process in the
same way as the original communication dependencies, we ensure that
states with an invalid global ordering cannot be reached.

We implement this more precise definition of a free task by maintaining
a record of $G_{S^{\prime}}^{*}$with each state in the form of a
transitive closure matrix. Whenever a new task is ordered, the transitive
closure is updated to reflect the new ordering dependencies. We can
then use this matrix to determine which tasks are free when creating
the children of a state.

\subsection{Evaluation\label{subsec:Invalid-States-Evaluation}}

Use of the simpler definition of the ready list as in Section~\ref{subsec:Ordering}
does not produce incorrect results; invalid states eventually have
their f-values escalate indefinitely, and hence are quickly removed
from consideration. An indeterminate amount of work was wasted before
reaching and removing these obviously wrong states, however. We therefore
find it necessary to evaluate whether the work saved by avoiding invalid
states outweighs the additional algorithmic overhead necessary to
do so. To determine experimentally the impact of this, we performed
A{*} searches on a set of task graphs using versions of the model
both with invalid state avoidance and without. Task graphs were chosen
corresponding to a wide variety of program structures. Approximately
270 graphs with 21 tasks were selected. These graphs were a mix of
the following DAG structure types: Independent, Fork, Join, Fork-Join,
Out-Tree, In-Tree, Pipeline, Random, Series-Parallel, and Stencil.
We attempted to find an optimal schedule using both 2 and 4 processors,
once each for both versions of AO, giving a total of over 1000 trials.
All tests were run on a Linux machine with 4 Intel Xeon E7-4830 v3
@2.1GHz processors. The tests were single-threaded, so they would
only have gained marginal benefit from the multi-core system. The
tests were allowed a time limit of 2 minutes to complete. For all
tests, the JVM was given a maximum heap size of 96 GB.A new JVM instance
was started for every search, to minimise the possibility of previous
searches influencing the performance of later searches due to garbage
collection and JIT compilation.

Figure \ref{fig:invalid-comparison} shows the results of these tests.
We use a form of plot know as a performance profile: the $x$-axis
shows time elapsed, while the $y$-axis shows the cumulative percentage
of problem instances which were successfully solved by this time.
Both versions of AO were able to solve approximately 70\% of the problem
instances within 2 minutes, with a slight advantage for invalid state
avoidance. This suggests that the presence of invalid states does
not have too much of a negative impact on average. It also suggests,
however, that the addition of the transitive closure and associated
operations does not significantly slow down the implementation of
the AO formulation.

\begin{figure}
\begin{centering}
\includegraphics[width=0.5\columnwidth]{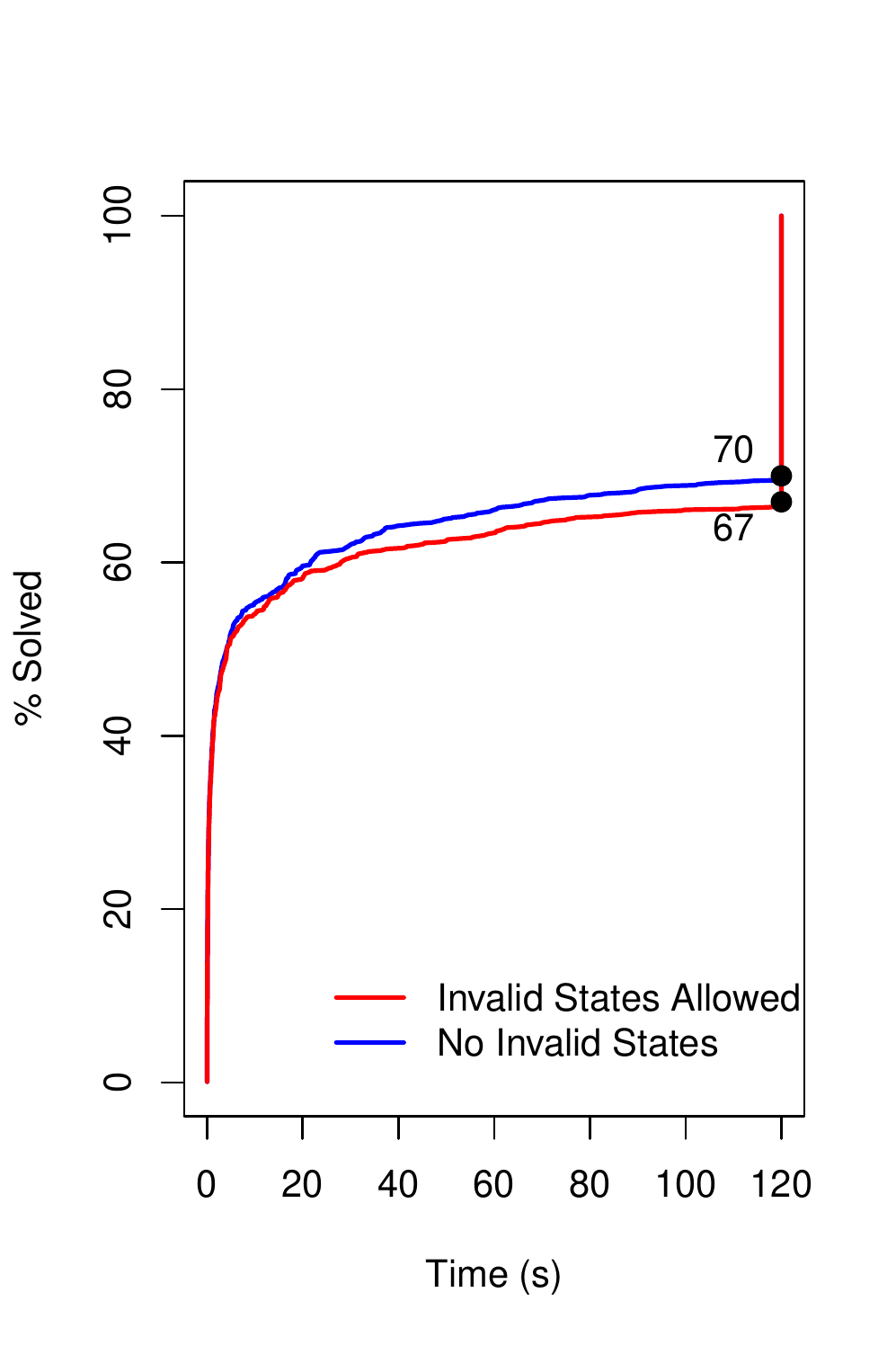}
\par\end{centering}
\caption{Comparing the performance of AO with and without invalid states.\label{fig:invalid-comparison}}
\end{figure}

\section{Pruning Techniques and Optimisations}

\label{sec:pruning}

Now that the novel AO model and the search through its solution space
have been proposed, it is important to investigate pruning techniques
and other optimisations which can be used with it. A search of the
AO state-space model is theoretically able to benefit from several
pruning techniques and optimisations already developed for ELS. Namely,
these are identical task pruning,  fixed order pruning and a heuristic
upper bound \citep{Sinnen2014201,ShaSinAst2010}. We discuss them
in the following, see which have become obsolete and then propose
a new additional pruning technique.

\subsection{Adapted from ELS}

\subsubsection*{Identical Task}

Two tasks $A$ and $B$ are considered identical if they are indistinguishable
from each other in any way except by their name \citep{SinnenElsevier2014}.
This means they have the same weight, same children, same parents,
and same communication costs to and from those respectively. If tasks
are identical, then their positions in any schedule can be freely
swapped without any effect on the rest of the schedule. Therefore,
the relative ordering of a set of such tasks, $I$, does not matter.
We only need to consider one such order for our search. 

To do this in ELS, we augment the task graph by creating a chain of
``virtual edges'' linking tasks in $I$. A virtual edge from task
$A$ to task $B$ prevents $B$ from being considered for scheduling
before $A$, but does not imply any other dependency. It has no weight,
and $A$ does not have to finish before $B$ can be started. Virtual
edges, therefore, only have an impact on deciding which tasks belong
to the current ready list. $B$ cannot be added to the ready list
until $A$ is scheduled. In this way it is ensured that only one order
for the identical tasks in $I$ is allowed.

In AO, this pruning can also be applied, but we need to distinguish
the allocation (A) and the ordering (O) phase. In the allocation phase
we take advantage of identical tasks to provide additional pruning
. Not only does the ordering of identical tasks not matter, a task
$A$ on processor $p_{1}$ can be swapped with its identical task
$B$ on $p_{2}$ without consequence. The processor that an individual
task in a set of identical tasks is allocated to does not matter.
What matters, therefore, is what number of the tasks in $I$ belong
to each part of the allocation. While building our allocation, we
give the parts of the allocation an arbitrary index order. Say that
part $a_{x}\in A$ is the part with highest index to which any task
in $I$ is allocated. When a task $n_{i}\in I$ is next to be allocated,
we restrict the parts to which it may be assigned to only those $a_{i}\in A|i\geq x$.
Essentially, as we allocate the tasks in $I$, we first decide how
many of the identical tasks are assigned to $a_{0}$, then how many
are assigned to $a_{1}$, and so on. In this way, we avoid producing
any allocations which differ only in the permutation of identical
tasks across processors.

During the ordering phase the pruning technique is applied in much
the same way as in ELS. However, since we consider only local ready
lists, the virtual edges are not always relevant. If identical tasks
$A$ and \textbf{$B$ }are allocated to the same processor, then the
virtual edge from $A$ to $B$ will be respected, and only orderings
in which $A$ goes before $B$ will be produced. If $A$ and \textbf{$B$
}are on different processors, however, then the order in which they
are considered for scheduling is instead decided by the order in which
processors are considered. These virtual edges will therefore be ignored. 

\subsubsection*{Heuristic Upper Bound}

A solution to a task scheduling problem can be found quickly (that
is, in polynomial time) using a heuristic algorithm. The length of
an approximate solution can then be used as an upper bound for f-values
in an optimal search. Since we have an example of a solution with
this length, it is guaranteed that no solution with a higher f-value
can be optimal. The A{*} algorithm will not examine states with higher
f-values than the optimal solution, so this optimisation will not
prevent additional states from being created. However, it can save
memory as states that will never need to be examined do not need to
be stored. Additionally, if the heuristic happens to find an optimal
solution, it saves us from searching through an indeterminate number
of equal f-value states. As soon as our best state has that f-value,
we know we can stop and take the heuristic solution - now proven to
be optimal \citep{SinnenElsevier2014,ShaSinAst2010}.

The AO model is able to use this technique in just the same way as
ELS, with no special consideration required.

\subsubsection*{Fixed Task Order}

For a fork graph, it is guaranteed that an optimal schedule exists
in which the tasks are scheduled in order of non-decreasing in-edge
weight \citep{Sinnen2014201}. Given such an order for the tasks,
all that is required to find the optimal schedule is a search for
an optimal allocation of tasks. Conversely, for join graphs, the same
is true if the tasks are scheduled in order of non-increasing out-edge
weight. In a fork-join graph, if an order can be found which satisfies
both the fork and the join condition simultaneously, then this order
is also optimal. We can therefore fix the order of this set of tasks,
eliminating the need to search all permutations. 

In fact in ELS, whenever the current set of ready tasks fulfills these
conditions, their order can be fixed \citep{Sinnen2014201}. This
fixed order can then be followed until changes to the ready list invalidate
the fixed order conditions. This property allows the technique to
be extended to graphs which are not purely independent, a fork, join,
or fork-join, but merely contain these as a sub-structure. If in any
state the current ready tasks all belong to such a sub-structure,
it may be possible to fix their order. It also means that if a fork-join
graph does not immediately meet the fixed order conditions, the conditions
may still be met in subsequent states and allow the order to be fixed
for certain sub-trees.

For AO, if we are able to fix the order of all the tasks in a graph
then only the allocation phase of our search is relevant, and the
ordering phase becomes trivial. Generally, in the ordering phase,
we can apply our fixed order conditions to the tasks in a local ready
list. If the conditions hold for the currently ready tasks on a given
processor, then the local order of these tasks can be fixed. Since
these local ready lists are smaller than the global ready list of
ELS, they are less likely to contain a task which contradicts the
fixed order conditions. Additionally, the number of chances for an
order to be fixed increases by a factor of $|P|$. It is therefore
probable that the order of tasks can be fixed more often when using
AO, albeit for smaller sets of tasks. We can even fix the order of
multiple fork-join substructures in a graph at the same time, as long
as they are assigned to different processors.

\subsection{Obsolete}

Since there are no duplicates in the AO model, the pruning techniques
used to mitigate the impact of duplicates in ELS \citep{Orr2015:dfs}are
no longer relevant.

\subsubsection*{Duplicate Detection}

In ELS, there is no way to mitigate the effect of independent scheduling
order duplicates (Section~\ref{subsec:Exhaustive-List-Scheduling})
except to keep a record of all states found so far. Binary search
trees are used to store both an open set (containing states created
but not yet explored) and a closed set (containing states which have
already been expanded) \citep{ShaSinAst2010}. Each time a state is
created, it must be ensured that neither the open nor the closed set
already contain this state. If they do, then this state is a duplicate
and is discarded. The need for the closed set means that all states
created must be kept in memory for the duration of the search. The
time taken to search the open and closed sets is $O(lg(k))$, where
$k$ is the number of states so far created by the search.

Since there are no duplicates in AO, it is not necessary to maintain
a closed list. It is also unnecessary to compare new states against
the open list, which permits the use of more efficient data structures
for the open list. AO therefore has a very large advantage in terms
of memory usage, and a small but perhaps practically significant advantage
in time complexity.

\subsubsection*{Processor Normalisation}

In order to avoid processor permutation duplicates (Section~\ref{subsec:Exhaustive-List-Scheduling})
in ELS, each state created has its partial schedule $S^{\prime}$
transformed into a normalised form $S_{N}^{\prime}$ \citep{ShaSinAst2010}.
In the ELS implementation compared here, we rename and therefore 're-order'
the processors to which tasks are assigned in $S^{\prime}$, in a
way that ensures that all processor permutation duplicates of $S^{\prime}$
are transformed to the same normalised form. First, we give a total
ordering to the tasks in $V$. This can be any order, so long as it
is used consistently, but is likely to be the same topological order
used elsewhere. We can then define $min(p_{i})$ as the task $n$
with lowest value among those assigned to $p_{i}$. The ordering of
processors is then defined such that $min(p_{i})<min(p_{j})\implies p_{i}<p_{j}$.
Once the processors are re-ordered according to this scheme, we have
our normalised partial schedule $S_{N}^{\prime}$. By normalising
all states in this way, we can use the previously discussed duplicate
detection mechanism in order to remove processor permutation duplicates
from consideration.

The method used by the AO model for its allocation phase makes this
process unnecessary, as it simply does not allow these duplicates
to be produced in the first place. In essence, the method of iteratively
building a partition from an ordered list of tasks ensures that each
state produced is automatically in its normalised form. In other words,
if the processor normalisation process were to be applied to one of
AO's partial allocations no re-ordering would ever take place, as
the processors are in their normalised order at all times.

\subsection{Novel}

\subsubsection*{Graph Reversal}

Among the standard task graph structures, there are several pairs
which differ only by the direction of their edges. Most obviously,
reversing the edges of a fork graph produces a join graph, and vice
versa. It can also be observed that, for both the ELS and AO models,
join graphs are significantly more difficult to solve than fork graphs,
as evident in Figure \ref{fig:graph-type-breakdown}. This arises
from the fact that in both models tasks are allocated to processors
in a topological order. In a fork graph, all communications originate
from the source task. Being a source, and therefore first in topological
order, this task is always the first to be allocated, and from that
point on it will always be known whether a communication cost is incurred
or not as soon as the other corresponding task is allocated. Conversely,
in a join graph, all communications go to the sink task. Since it
is last in topological order, all other tasks are allocated before
it without any knowledge of which communications are incurred. The
decision of which communication costs to set to zero is made all at
once, in the last step. Since knowledge about communication costs
is critical to determining a lower bound on the eventual length of
a partial schedule, it is clear to see why fork graphs, where this
information is always available, can be solved much more efficiently
than join graphs, where this information is not available until too
late.

We define $R(G)$, the reverse of a task graph $G$, simple by reversing
the direction of each edge $e\in G$. If we do this for a join graph,
the result is a corresponding fork graph. We can then find an optimal
schedule $S_{R(G)}^{*}$ for this reversed graph. If this reversed
graph is now a fork, it can be solved much more efficiently than the
original join. Our optimal schedule for $R(G)$ can then be reversed
to obtain a valid schedule for $G$. Reversing a schedule simply means
reversing the ordering of the tasks allocated to each processor, so
that the task scheduled first is now scheduled last, and so on. Using,
this ordering, each task is started as early as possible, subject
to the processor availability and precedence constraints. This reversed
schedule $R(S_{R(G)}^{*})$ is now an optimal schedule $S_{G}^{*}$
for the original graph. This known result is easy to show. Taking
a schedule and reversing its time line, while at the same time reversing
the direction of all edges results in a valid schedule for the reversed
graph. The reversed schedule and the original must have the same length
- the same computation and communication delays occur, only in backwards
order. In the reversed schedule, not all tasks might start at their
earlist possible time, but rescheduling them earlier has no negative
impact on the schedule length (this is done automatically in the above
described procedure as we only take the order of the tasks). If a
shorter schedule for $G$ existed, it could in turn be reversed to
produce a shorter valid schedule for $R(G)$, and so $S_{R(G)}^{*}$
could not have been an optimal schedule in the first place. Instead
of solving difficult join graphs, we can instead transform them into
fork graphs, solve these much more easily, and then transform the
resulting schedules to produce optimal schedules for the original
joins. This technique is also applied to out-tree and in-tree graphs,
of which fork and join are special cases, respectively.

\section{Evaluation}

\label{sec:evaluation}

In this section we evaluate the benefit of the new AO state-space
model in the search for optimal solutions. For this purpose it is
compared against the use of the ELS model. The empirical evaluation
was performed by running branch-and-bound searches on a diverse set
of task graphs using each state-space model. Task graphs were chosen
that differed by the following attributes: graph structure, the number
of tasks, and the communication-to-computation ratio (CCR). Table
\ref{tab:Range-of-task} describes the range of attributes in the
data set. A set of 1360 task graphs with unique combinations of these
attributes were selected. These graphs were divided into four groups
according to the number of tasks they contained: either 10 tasks,
16 tasks, 21 tasks, or 30 tasks. An optimal schedule was attempted
for each task graph using 2, 4, and 8 processors, once each for each
state-space model. This made a total of 4080 problem instances attempted
per model. Searches were performed using the A{*} search algorithm.
All pruning techniques discussed in the previous section were applied
to each state-space model that could take advantage of them.

\setlength{\columnsep}{1.5cm} \setlength{\columnseprule}{0.4pt}

\begin{table}
\begin{multicols}{3}\raggedcolumns 
\begin{centering}
\textbf{Graph Structure}
\par\end{centering}
\begin{itemize}
\item Independent
\item Fork
\item Join
\item Fork-Join
\item Out-Tree
\item In-Tree
\item Pipeline
\item Random
\item Series-Parallel
\end{itemize}
\columnbreak
\begin{centering}
\textbf{No. of Tasks}
\par\end{centering}
\begin{itemize}
\item 10
\item 16
\item 21
\item 30
\end{itemize}
\columnbreak
\begin{centering}
\textbf{CCR}
\par\end{centering}
\begin{itemize}
\begin{singlespace}
\item 0.1
\end{singlespace}
\item 1
\item 10
\end{itemize}
\end{multicols}

\caption{\label{tab:Range-of-task-1}Range of task graphs in the experimental
data set.}
\end{table}
The implementations were built with the Java programming language.
An existing implementation of ELS\citep{Sinnen2014201} was used as
the basis for an AO implementation, with code for common procedures
shared wherever possible. Notably, the basic implementation of the
A{*} search algorithm is shared, with the implementations differing
only by how the children of a search node are created. The implementations
of commonly applicable pruning techniques are also shared. Using this
approach, the differences observed in the experimental results are
most like due to the different models and not implementation artifacts.

All tests were run on a Linux machine with 4 Intel Xeon E7-4830 v3
@2.1GHz processors. The tests were single-threaded, so they would
only have gained marginal benefit from the multi-core system. The
tests were allowed a time limit of 2 minutes to complete. For all
tests, the JVM was given a maximum heap size of 96 GB. A new JVM instance
was started for every search, to minimise the possibility of previous
searches influencing the performance of later searches due to garbage
collection and JIT compilation.

\subsection{Results and discussion}

For the 10-task group, every problem instance was solved within at
most 3 seconds, regardless of the state-space model used. It is apparent
that both models are powerful enough that task graphs this small will
not present a challenge, and so we will not discuss these results
further. The other three groups will be discussed in detail in the
following.. Figure~\ref{fig:overall-results} shows performance profiles
(as used in Section \ref{subsec:Invalid-States-Evaluation}) that
compare the performance of the two models, broken down by graph size.
These charts indicate the accumulated percent of problem instances
in the data set that were successfully solved after a given time had
elapsed, up to the timeout of 120 seconds. In the 16 task group, a
large majority of the problem instances were able to be solved by
both models, but AO gives a clear advantage. By the timeout, 90\%
of instances were solved by AO, while only 77\% were solved by ELS.
In the 21 task group, we see an even more dramatic difference: while
ELS solves only 46\% of instances within two minutes, AO manages to
solve 70\%. Graphs of size 30 are difficult to solve within two minutes
for both models. Although ELS is seen to have a slight advantage in
the first few seconds of runtime, by the end of two minutes the lines
have converged and both models solve 19\% of instances. Overall, AO
has significantly better performance in these experiments, particularly
in the ``medium difficulty'' 21 task group. Not only does it solve
more instances within a few seconds, the gap between the models widens
as time goes on, with AO consistently solving more instances than
ELS. 

\begin{figure}
\begin{centering}
\includegraphics[width=0.75\columnwidth]{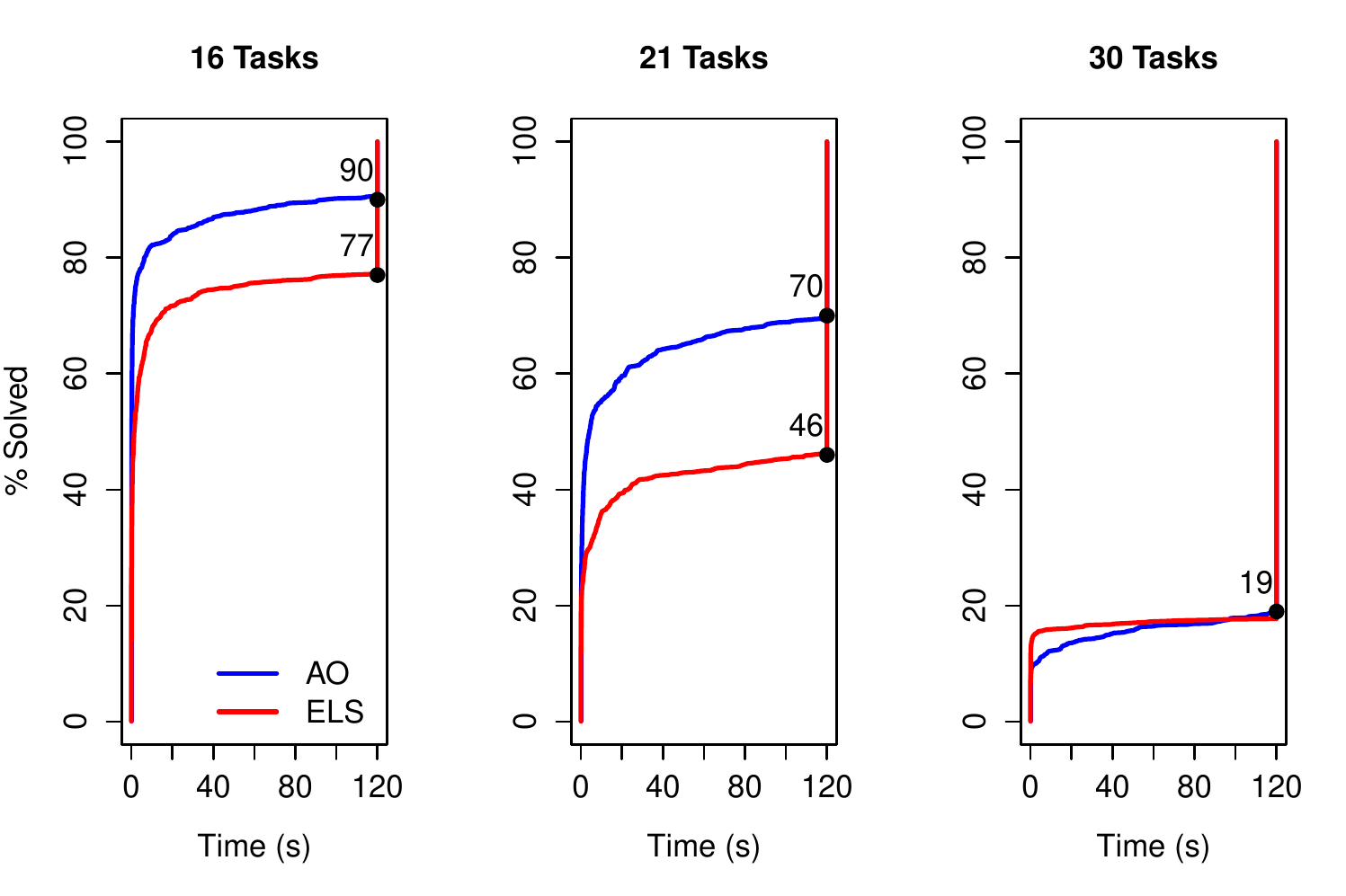}
\par\end{centering}
\caption{Overall performance of the two models.\label{fig:overall-results}}
\end{figure}
Breaking the results down by graph structure, we see that AO has a
clear advantage for most structures. Figure~\ref{fig:Performance-by-structure}
shows the solved instances within the time limit by model in stacked
bar charts across the different structures. In the 16 and 21 task
graphs, AO dominates ELS in almost all structures. For 30 tasks, it
is clear that graphs of this size present a significant challenge
for both models, as a large majority of problem instances were not
solved by either. Overall, both models solved 19\% of graphs in this
group. ELS shows better performance with Independent, Random, and
Stencil graphs, while AO is better for the other structures. The large
number of non-solved instances in this size category makes it difficult
to draw further conclusions. 
\begin{center}
\begin{figure*}[p]
\begin{centering}
\subfloat[16 Tasks]{\begin{centering}
\includegraphics[width=0.75\columnwidth]{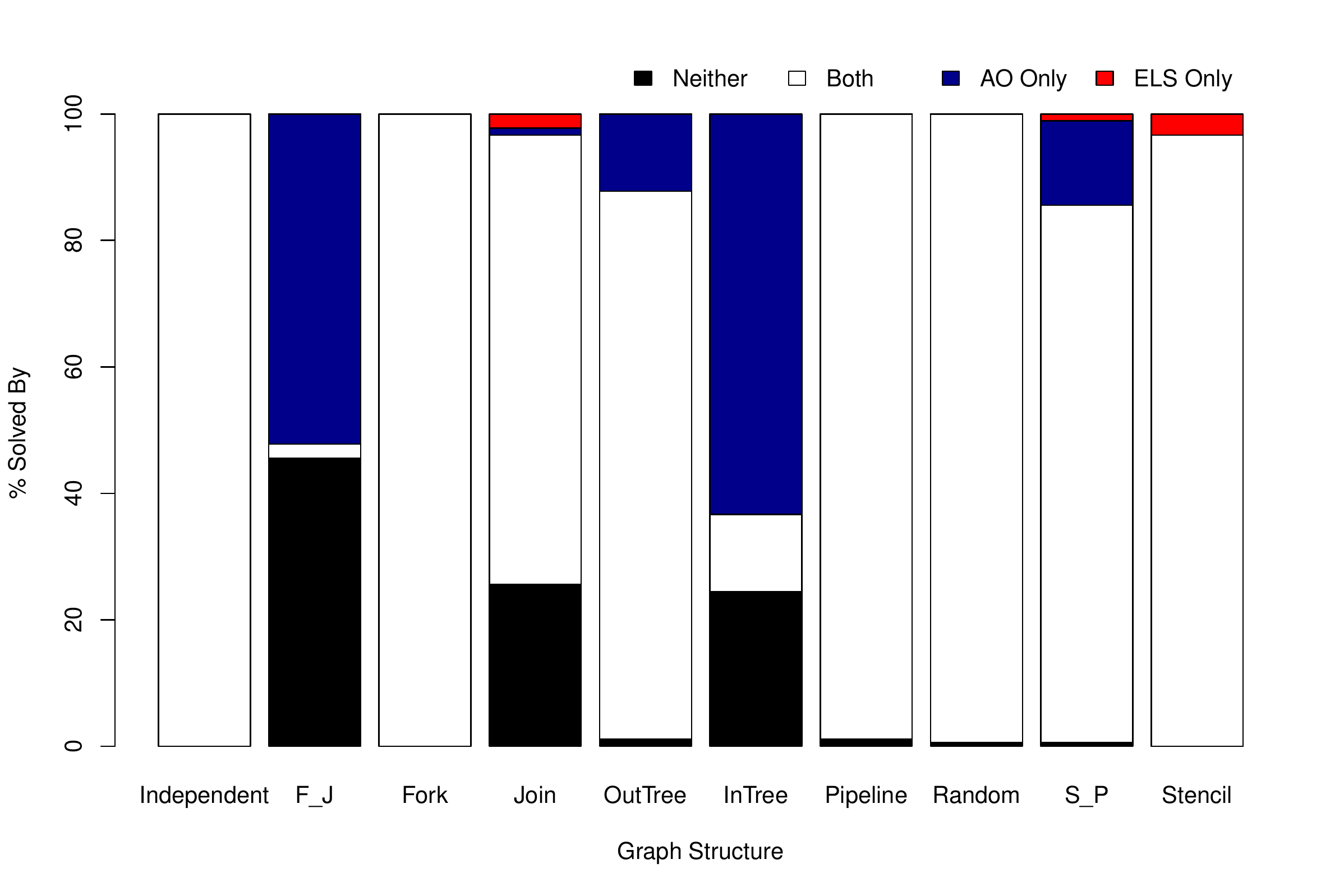}
\par\end{centering}
}
\par\end{centering}
\begin{centering}
\subfloat[21 Tasks]{\begin{centering}
\includegraphics[width=0.75\columnwidth]{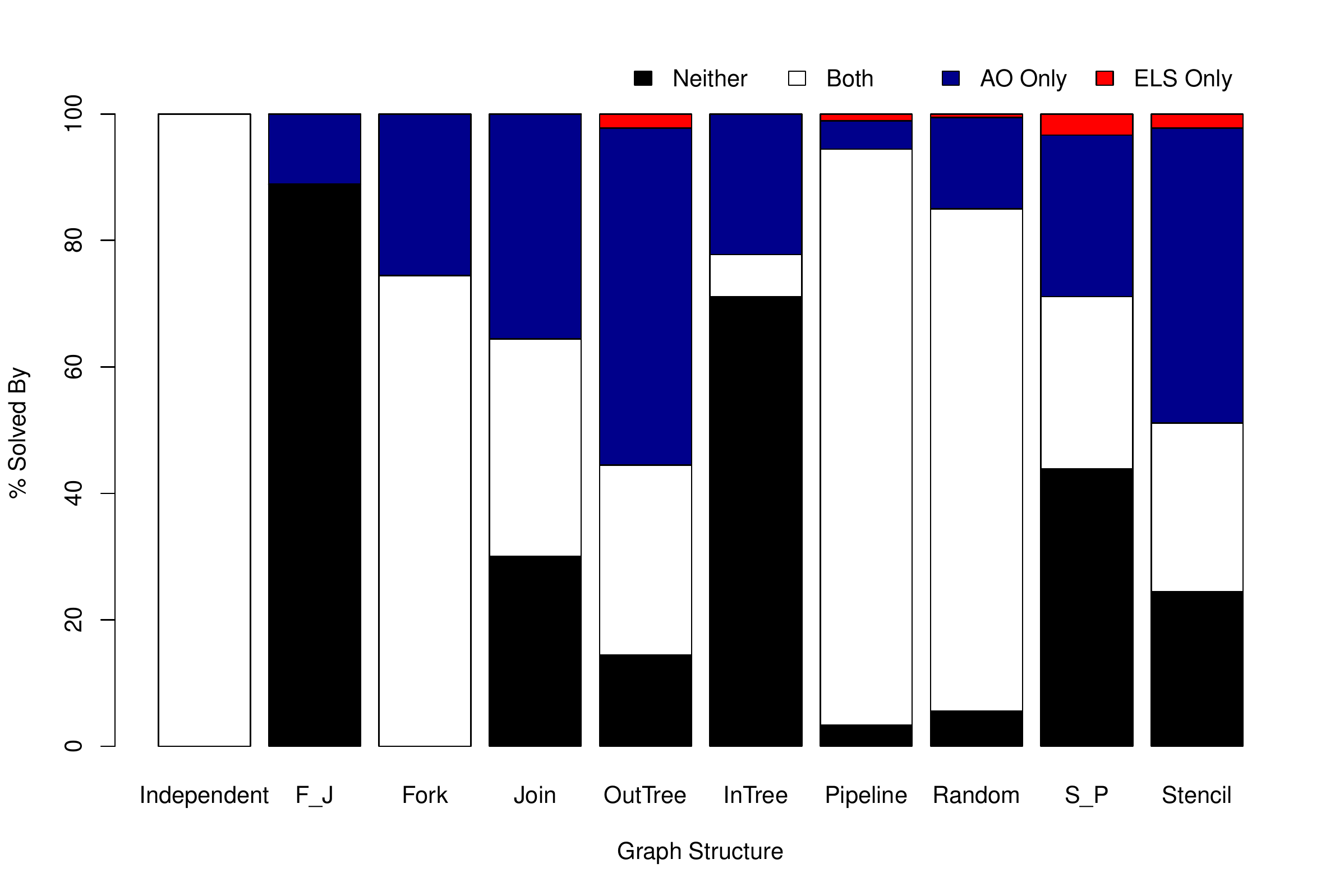}
\par\end{centering}
}\hfill{}\subfloat[30 Tasks]{\begin{centering}
\includegraphics[width=0.75\columnwidth]{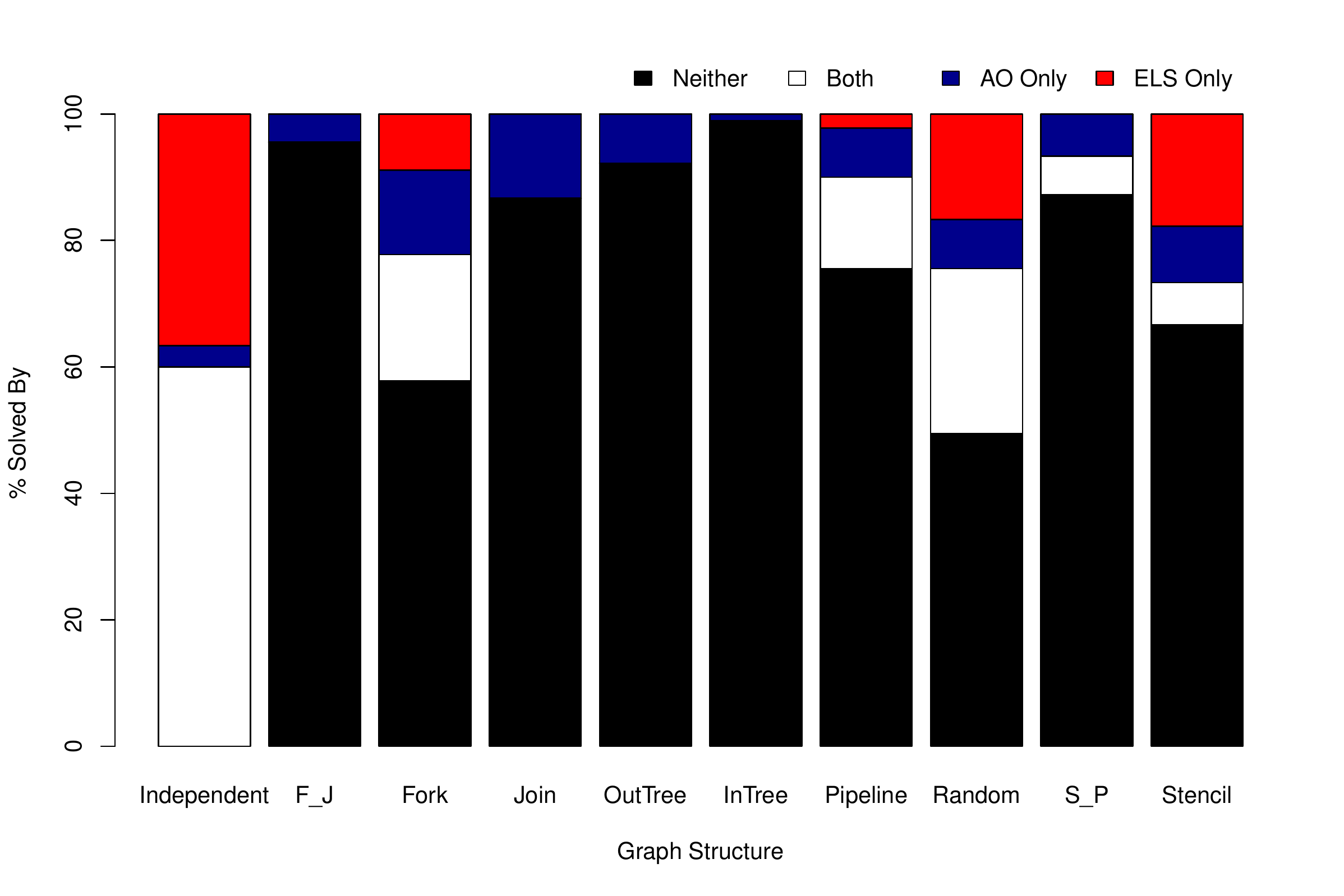}
\par\end{centering}
}
\par\end{centering}
\caption{\label{fig:Performance-by-structure}Performance of the models broken
down by graph structure.}
\label{fig:graph-type-breakdown}
\end{figure*}
\par\end{center}

Comparing by the communication-to-computation ratio of the task graphs
(Figure~\ref{fig:CCR-results-1}), we see that AO has an advantage
at all values, but is dramatically better at solving graphs with very
high CCR of 10. By deciding the allocation of tasks first, a search
using the AO model very quickly determines the entire set of communication
costs which will be incurred. Allocations which incur very large communication
costs are likely to be quickly ruled out, and knowledge of all the
communication costs can be used in the calculation of $f$-values
throughout the ordering stage. For graphs in which communication is
dominant, it is intuitive that early knowledge of the communication
would allow more efficient decision-making, and these results support
that intuition.

\begin{figure}
\begin{centering}
\includegraphics[width=0.75\columnwidth]{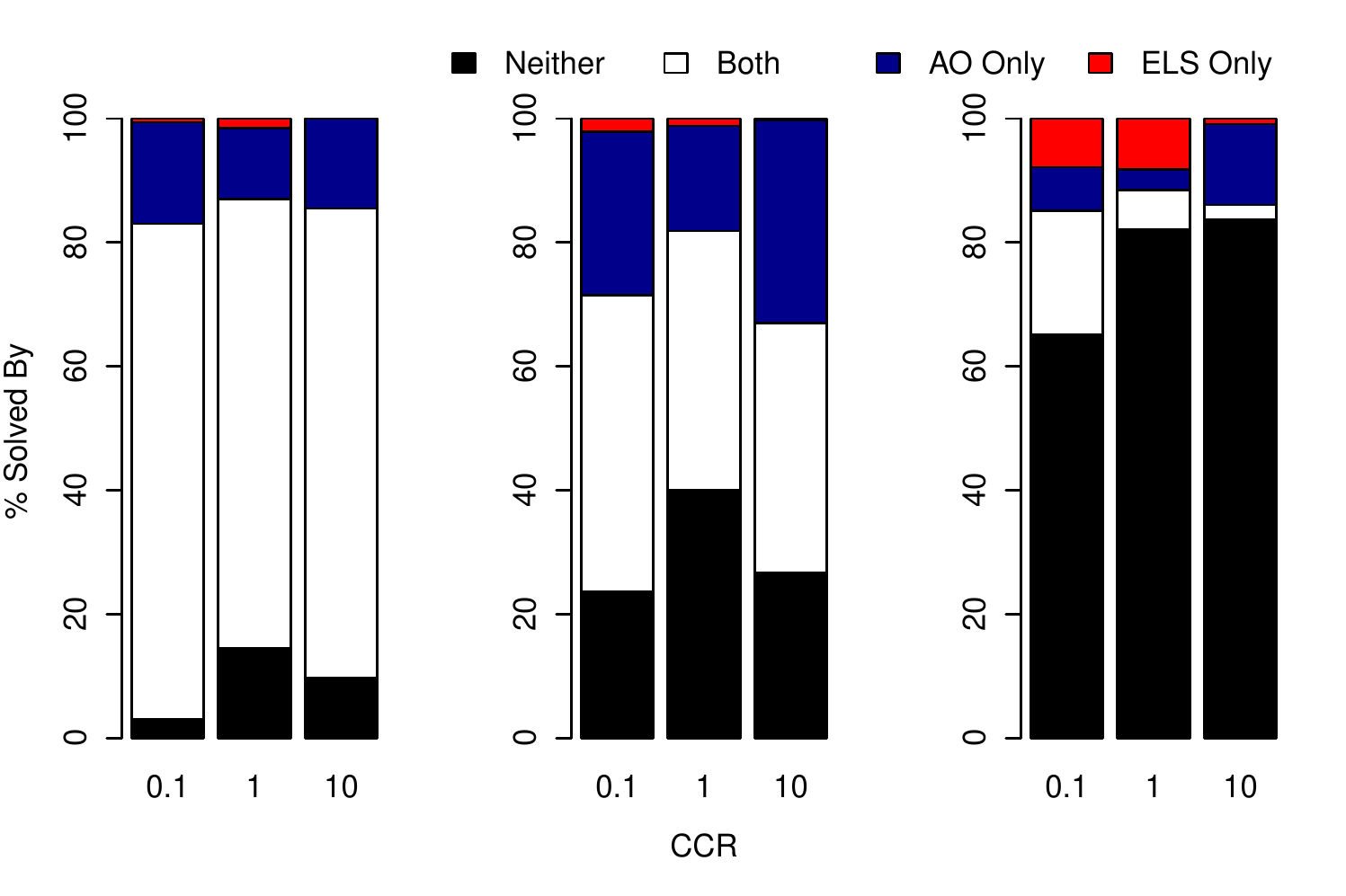}
\par\end{centering}
\caption{\label{fig:CCR-results-1}Performance of the models broken down by
CCR.}
\end{figure}
The newly implemented reversed-join pruning technique allowed for
a dramatic increase in the number of join and in-tree task graphs
able to be solved, as seen in Figure \ref{fig:reverse-structure},
where the '-R' graph name extension indicates that the graph was reversed
before scheduling. As would be expected, the performance trend for
reversed join graphs is very similar to that for fork graphs, and
the same is true for reversed in-tree and out-tree graphs. Both AO
and ELS are able to solve many more graphs this way, but the impact
is more significant for AO. 

\begin{figure}
\begin{centering}
\includegraphics[width=0.75\columnwidth]{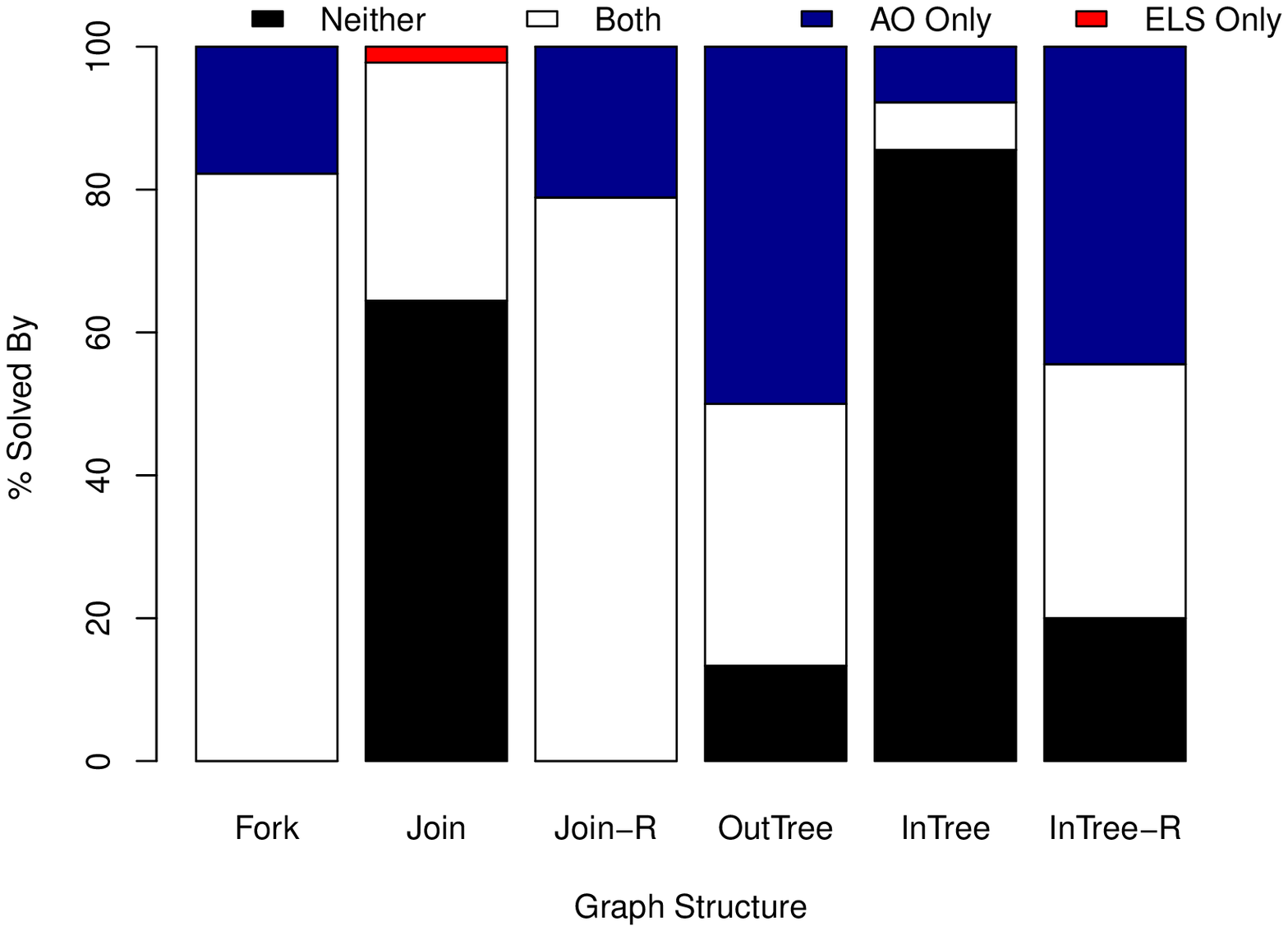}
\par\end{centering}
\caption{The effect of reversing Join and In-Tree task graphs.}
\label{fig:reverse-structure}
\end{figure}

\section{Advanced Lower Bound Heuristics\label{sec:Lower-Bound-Heuristics}}

Section \ref{sec:theory} described intuitive and essential bounds
and corresponding $f$-value functions to use for the initial implementation
of the AO state space. However, it is likely that tighter bounds could
improve the performance of searches using AO. Below we describe various
improvements which allow these bounds to be tightened in some circumstances.
These new bounds all apply to the Allocation phase of the state-space.
Since this is the part of the model which is most distinct from ELS,
it appears to offer the most new opportunities for heuristic development.

\subsection*{Minimum Finish Time}

The allocation-load heuristic (Section~\ref{subsec:Allocation-Cost-Function},$f_{load}$\eqref{eq:fload-bl})
was intended to provide an estimate for the finish time of each processor,
using the trivial fact that a processor will require at least as much
time as is necessary to execute all of its allocated tasks. The addition
of the minimum allocated top and bottom levels was a simple step which
acknowledged that processors must sometimes wait for data to be communicated
to them before they can begin execution. However, it is also often
the case that a processor cannot execute all of its tasks in one unbroken
stretch, and additional idle time must occur. Idle time is time during
the run of a schedule when a processor is not performing any computation.
This is often necessitated when a processor has finished execution
of one task but has not yet received all data necessary to begin the
next task in order. It is required to wait while computation and communication
associated with that task's ancestors is performed by other components
of the system. The time at which a processor receives all the necessary
data for a task $n$ and may begin execution is known as the data-ready
time, $drt(n)$. In any schedule, for a given task $n$, it must always
be the case that $t_{s}(n)\geq drt(n)\geq tl_{\alpha}(n)$. It follows
from this that we may obtain a tighter bound by considering the allocated
top levels of all tasks on a processor, thereby including some additional
idle time.

Consider a single grouping $a$ in a partial partition $A$. When
scheduled, each task $n\in a$ will have some finish time $t_{f}(n)=t_{s}(n)+w(n)$.
We wish to find a lower bound for the value $T_{F}(a)=max_{n\in a}\{t_{f}(n)\}$.
A simple observation is that $T_{F}(a)\geq max_{n\in a}\{tl_{\alpha}(n)+w(n)\}$:
the minimum finish time must be at least as large as the maximum of
the earliest possible finish times of the tasks. However, in many
cases a task $n$ will not be able to start at $tl_{\alpha}(n)$ without
overlapping execution of tasks and violating the processor constraint.
Execution of the final task must be delayed until all other tasks
are completed. 

Given a total ordering $O_{a}$ for the tasks $n\in a$, let task
$n_{i}$ be the i\textsuperscript{th} in order according to $O_{a}$.
The earliest possible starting time for task $n_{i}$ adhering to
order $O_{a}$ is then $t_{est}(n_{i},O_{a})=max\{tl_{\alpha}(n_{i}),t_{est}(n_{i-1},O_{a})+w(n_{i-1})\}$,
with $t_{est}(n_{0},O_{a})=tl_{\alpha}(n_{0})$. It follows for the
order $O_{a}$ that $T_{F}(a,O_{a})\geq max_{n\in a}\{t_{est}(n,O_{a})+w(n)\}.$

To use this in a lower bound, we now need to find an ordering which
minimises $T_{F}(a)$.
\begin{lem}
Let $O_{a}^{tl_{\alpha}}$ be the ordering of the tasks $n\in a$
in non-descending allocated top level order $tl_{\alpha}$. Then $T_{F}(a,O_{a}^{tl_{\alpha}})\leq T_{F}(a,O_{a})$
for all possible orders $O_{a}$ of tasks in $a$. 
\end{lem}
\begin{proof}
The proof is by contradiction. Assume there is an ordering$O_{a}^{*}\ne O_{a}^{tl_{\alpha}}$
with $T_{F}(a,O_{a}^{tl_{\alpha}})>T_{F}(a,O_{a}^{*})$. Now consider
two adjacent tasks $n_{j-1}$ and $n_{j}$, $1\leq j\leq|a|-1$, in
the ordering $O_{a}^{*}$. If $tl_{\alpha}(n_{j})\geq tl_{\alpha}(n_{j-1})$,
they are already in the same relative order as in $O_{a}^{tl_{\alpha}}$.
Otherwise, their order can be swapped without increasing the earliest
start time of the following tasks, $t_{est}(n_{i}),\forall i>j$.
This is true because $tl_{\alpha}(n_{j})<tl_{\alpha}(n_{j-1}),$and
so there is no gap (idle time) between the two tasks. This means they
can be swapped and placed again without any gap, such that the now
first task $n_{j}$ starts at the original $t_{est}(n_{j-1}O_{a}^{*})$.
Following this, the now later task $n_{j-1}$ will finish at the original
$t_{f}(n_{j},O_{a}^{*})$, and therefore the swap does not impact
on the following tasks in the order. Repeatedly swapping all task
pairs that are not in order until there is no such pair left, will
then bring all tasks of $a$ into $O_{a}^{tl_{\alpha}}$ order with
the same $T_{F}(a,O_{a}^{*})$ as before, which is a contradiction
to the original assumption.
\end{proof}
We have shown that by ordering the tasks $n\in a$ by non-descending
allocated top level, we can obtain a lower bound for $T_{F}(a)$,
which we call $T_{F}^{*}(a)$. Adding the minimum bottom level to
this, as explained in the previous section, achieves an even tighter
lower bound for the overall schedule. 

Additionally, an analogous bound $T_{S}^{*}(a)$ can be found by arranging
the tasks in an order $O_{a}^{bl}$, such that $bl(n_{i})-w(n_{i})\geq bl(n_{i+1})-w(n_{i+1})$,
and assigning them as-late-as-possible start times. Note that this
is equivalent to the process of finding $T_{F}^{*}(a)$ on a reversed
task graph. Where $T_{F}^{*}(a)$ is a lower bound on the time between
the start of the overall schedule and the end of execution of the
last task on the processor, $T_{S}^{*}(a)$ is a lower bound on the
time between the start of execution of the first task on the processor
and the end of the overall schedule. This can of course be combined
with the minimum allocated top level in just the same way as $T_{F}^{*}(a)$
is with the minimum bottom level. Finally, we can combine all of this
to obtain our new overall bound:

\begin{align}
f_{\mathrm{load-mft}}(s)=\mathrm{max}{}_{a\in A} & \biggl\{ max\{T_{F}^{*}(a)+min_{n\in a}(bl_{\alpha}(n)-w(n)),\nonumber \\
 & min_{n\in a}tl_{\alpha}(n)+T_{S}^{*}(a)\biggr\}\label{eq:fload-mft}
\end{align}

\subsection*{Critical Path Load}

The allocated critical path heuristic (Section~\ref{subsec:Allocation-Cost-Function},
$f_{acp}$\ref{eq:facp}) can be improved by closer examination of
the reasons for using top and bottom levels. The top level $tl(n)$
gives us a lower bound on the time before a task $n$ can start. The
bottom level $bl(n)$ gives us a lower bound on the time between the
start of task $n$ and the overall finish time of the schedule. These
are determined by finding the \textquotedbl critical paths\textquotedbl{}
in the task graph that begin and end with that task, respectively.
However, as shown by use of the allocation-load heuristic, combining
a critical path with a load balancing heuristic gives us a tighter
bound. To apply this to top and bottom levels, we can examine the
ideal load balancing of the tasks preceding and following our task
$n$. We call these the top load, $tload(n)$, and bottom load, $bload(n)$.
We use a well-known simple load-balancing bound found by summing the
weights of all relevant tasks and dividing by the total number of
processors. For the top load, this means the sum of the weights of
all ancestors of our task: since they all must finish execution before
our task starts, a perfect load balancing of these tasks gives a lower
bound for the start of our task. Therefore, $tload(n)=\sum_{i\in ancestors(n)}w(i)/|P|$.
Similarly, the bottom load uses the sum of the weights of all descendants
of our task: since they all must start execution after our task finishes,
a perfect load balancing gives a lower bound for the time required
to finish the schedule, and so $bload(n)=\sum_{i\in descendants(n)}w(i)/|P|$.
The top and bottom load values are not affected by allocation, and
therefore need only be calculated once at the beginning of the search
process. When determining the allocated critical path, we can use
whichever is the maximum of the top level or top load of each task
(and similarly, bottom level or bottom load) to find a tighter overall
bound: 

\begin{align}
f_{\mathrm{acp-load}}(s)=\mathrm{max}_{n\in V^{\prime}} & \biggl\{ max(tl_{\mathrm{a}}(n),tload(n))\nonumber \\
 & +max(bl_{a}(n),bload(n))\biggr\}\label{eq:facp-load}
\end{align}

\subsection{Evaluation}

To determine experimentally the impact of these new lower bound heuristics,
we performed A{*} searches on a set of task graphs using various heuristic
profiles, detailed in Table \ref{tab:Heuristic-profiles}. To clearly
see the effect of these novel lower bounds, we keep the A{*} search
algorithm fixed, except for the changes in the $f$-value calculation
using the proposed bounds. In each trial, the number of states created
by the A{*} search before reaching an optimal solution was recorded.
The same set of task graphs was used as in Section~\ref{sec:evaluation}.
We attempted to find an optimal schedule using 4 processors, once
each for each heuristic profile, giving a total of over 1000 trials.
As before, the algorithms were implemented in the Java programming
language. All tests were run on a Linux machine with 4 Intel Xeon
E7-4830 v3 @2.1GHz processors. The tests were single-threaded, so
they would only have gained marginal benefit from the multi-core system.
The tests were allowed a time limit of 2 minutes to complete. For
all tests, the JVM was given a maximum heap size of 96 GB.A new JVM
instance was started for every search, to minimise the possibility
of previous searches influencing the performance of later searches
due to garbage collection and JIT compilation.

\begin{table}
\begin{tabular}{c|c}
Heuristic Profile  & Description\tabularnewline
\hline 
Baseline & The previously used heuristics, described by Eqs. \eqref{eq:fload-bl}
and \eqref{eq:facp}.\tabularnewline
CritPathLoad & $f_{acp}$\eqref{eq:facp} is replaced by $f_{acp-load}$\eqref{eq:facp-load}.\tabularnewline
MinFinishTime & $f_{load}$\eqref{eq:fload-bl} is replaced by $f_{load-idle}$\eqref{eq:fload-mft}.\tabularnewline
\end{tabular}

\caption{\label{tab:Heuristic-profiles}Heuristic profiles used for experimental
trials.}
\end{table}
Figure \ref{fig:heuristic_profiles} shows the results of these tests
as performance profiles for the three heuristics. It is obvious that
the CPL heuristic produces no difference from the baseline, with both
solving 68\% of instances overall. This suggests that the top and
bottom load metrics rarely produce a critical difference when compared
to top and bottom level. However, MFT does show an advantage over
the baseline, with a total of 71\% of instances solved. The inclusion
of necessary idle time in the processor load heuristic is significant
enough to produce a difference of 3\% of total graphs solved. Further,
that advantage is achieved very early in the A{*} search, which makes
this even more useful. This noticeable improvement of the AO model
based A{*} search demonstrates the further potential that this new
AO model may show as better $f$ functions are discovered.

\begin{figure}
\begin{centering}
\includegraphics[width=0.5\columnwidth]{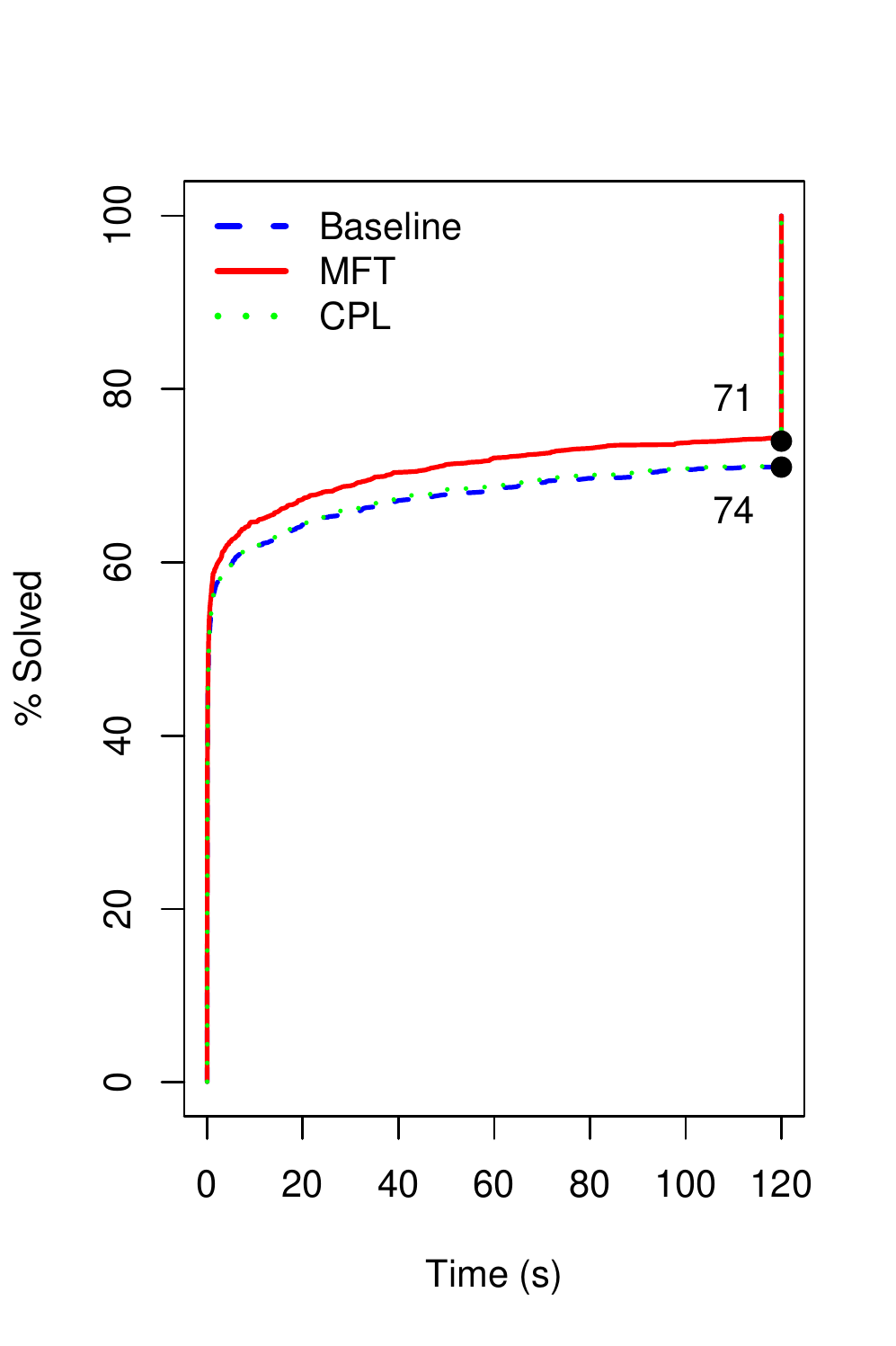}
\par\end{centering}
\centering{}\caption{Comparing percent of instances solved with different heuristic profiles.\label{fig:heuristic_profiles}}
\end{figure}

\section{Conclusions}

\label{sec:conclusion}

Previous attempts at optimal task scheduling through branch-and-bound
methods have used a state-space model which we refer to as exhaustive
list scheduling. This state-space model is limited by its high potential
for producing duplicate states. In this paper, we have proposed a
new state-space model which approaches the problem of task scheduling
with communication delays in two distinct phases: allocation, and
ordering. In the allocation phase, we assign each task to a processor
by searching through all possible groupings of tasks. In the ordering
phase, with an allocation already decided, we assign a start time
to each task by investigating each possible ordering of the tasks
on their processors. Using a priority ordering on processors, and
thereby fixing the sequence in which independent tasks must be scheduled,
we are able to avoid the production of any duplicate states. 

Our evaluation suggests that the AO model allows for superior performance
in the majority of problem instances, as searches using the new model
were significantly more likely to reach a successful conclusion within
a two minute time limit. Along with pruning techniques and optimisations
previously used with ELS, a new graph reversal optimisation leads
to much higher success rates with Join and In-Tree graph structures.
Finally, the more complex Minimum Finish Time heuristic for the allocation
phase of AO was found to produce a significant advantage over a simpler
heuristic.  In this work, the majority of pruning techniques for ELS
have either been adapted to AO, or made obsolete by its lack of duplicates.
However, entirely new pruning techniques which are only possible under
the AO model are likely to be able to be found. Application of these
new techniques may allow AO to further improve.

The AO state-space model's lack of duplicates means that branch-and-bound
algorithms searching it do not require the use of additional data
structures for duplicate detection. This suggests that AO may have
an advantage over ELS when using a low-memory algorithm such as depth-first
branch-and-bound. Parallel branch-and-bound is also likely to benefit
from the AO model, as the lack of duplicate detection means a decreased
need for synchronisation.

\bibliographystyle{elsarticle-num}
\bibliography{journalref}
 
\end{document}